\documentclass{lmcs}

\usepackage{hyperref}

\usepackage{microtype}
\usepackage[T5,T1]{fontenc}
\usepackage[utf8]{inputenc}

\usepackage{cleveref}
\usepackage{amsfonts, amsmath, amssymb, amsthm, dsfont, mathrsfs}
\usepackage{mathtools}
\usepackage{graphicx}
\usepackage{tikz}
\usepackage{tikz-cd}
\usepackage{csquotes}
\usetikzlibrary{shapes,arrows,positioning,automata,shadows}
\usepackage[ruled]{algorithm2e}
\usepackage{stmaryrd}
\usepackage{mathpartir}
\usepackage{soulutf8}
\usepackage{yhmath}
\usepackage{bussproofs}
\usepackage{cmll}

\usepackage{ifthen}

\newboolean{talk}
\setboolean{talk}{false}
\newboolean{paper}
\setboolean{paper}{false}

\newboolean{LMCSstyle}
\setboolean{LMCSstyle}{false}
\newboolean{IEEEstyle}
\setboolean{IEEEstyle}{false}
\newboolean{lipicsstyle}
\setboolean{lipicsstyle}{false}
\newboolean{eptcsstyle}
\setboolean{eptcsstyle}{false}
\newboolean{sigplanstyle}
\setboolean{sigplanstyle}{false}
\newboolean{PACMPL}
\setboolean{PACMPL}{false}
\newboolean{acmartstyle}
\setboolean{acmartstyle}{false}

\newboolean{needstheorems}
\setboolean{needstheorems}{false}
\newboolean{withimages}
\setboolean{withimages}{false}
\newboolean{withproofs}
\setboolean{withproofs}{true}

\newboolean{french}
\setboolean{french}{false}

\ifthenelse{\boolean{talk}}{}{\usepackage[utf8]{inputenc}}
\ifthenelse{\boolean{PACMPL}}{}{
	\ifthenelse{\boolean{french}}{\usepackage[french]{babel}}{
	  \ifthenelse{\boolean{lipicsstyle}}{}{\usepackage[english]{babel}}}
	  }
\ifthenelse{\boolean{PACMPL}}{}{
	\ifthenelse{\boolean{acmartstyle}}{}{
		\usepackage{amssymb} 
	}
}
\usepackage{amsmath}
\usepackage{graphicx}
\usepackage{fancybox}
\usepackage{cmll}
\usepackage{stmaryrd}
\usepackage{multirow}
\ifthenelse{\boolean{IEEEstyle}}{
\usepackage[bookmarks={false}]{hyperref}}{
\usepackage{hyperref}}
\usepackage{proof}
\usepackage{hhline}
\usepackage{xspace}
\usepackage{booktabs}
\usepackage{subcaption}
\usepackage{array}
\usepackage{tabularx}
\usepackage{arydshln}
\usepackage{commath}
\ifthenelse{\boolean{IEEEstyle}}{
\setboolean{needstheorems}{true}}{}

\ifthenelse{\boolean{eptcsstyle}}{
\setboolean{needstheorems}{true}}{}

\ifthenelse{\boolean{sigplanstyle}}{
\setboolean{needstheorems}{true}}{}

\newcommand{\myproof}[1]{
\ifthenelse{\boolean{withproofs}}{#1}{}}

\newcommand{\withproofs}[1]{
\ifthenelse{\boolean{withproofs}}{#1}{}}

\newcommand{\withoutproofs}[1]{
\ifthenelse{\boolean{withproofs}}{}{#1}}

\newcommand{\tm}{t}
\newcommand{\tmtwo}{u}

\newcommand{\var}{x}

\newcommand{\Rew}[1]{\rightarrow_{#1}}

\renewcommand{\to}{\Rew{}}

\newcommand{\tob}{\Rew{\beta}}

\newcommand{\ctxholep}[1]{[#1]}
\newcommand{\ctxhole}{\ctxholep{\cdot}}

\newcommand{\ctx}{C}
\newcommand{\ctxtwo}{D}
\newcommand{\ctxthree}{E}

\newcommand{\ctxp}[1]{\ctx\ctxholep{#1}}

\newcommand{\nbvctxtwo}[1]{\nbvctxtwo{#1}}

\newcommand{\grameq}{::=}

\newcommand{\isub}[2]{\{#1/#2\}}
\newcommand{\esub}[2]{[#1/#2]}

\newcommand{\llbrace}{\{ \kern -0.27em \vert}
\newcommand{\rrbrace}{\vert \kern -0.27em \}}

\ifthenelse{\boolean{LMCSstyle}}{}{
	
	}

\newcommand{\red}[1]{{\color{red} {#1}}}
\newcommand{\blue}[1]{{\color{blue} {#1}}}

\newcommand{\ignore}[1]{}

\newcommand{\myinput}[1]{\ifthenelse{\boolean{withimages}}{\input{#1}}{}}

\ifthenelse{\boolean{PACMPL}}{}{
	\ifthenelse{\boolean{acmartstyle}}{}{
		
	}
}

\newcounter{numberone}

\newcounter{numbertwo}

\renewcommand{\ctxholep}[1]{\langle #1\rangle}
\newcommand{\ctxtwop}[1]{\ctxtwo\ctxholep{#1}}
\newcommand{\ctxthreep}[1]{\ctxthree\ctxholep{#1}}

\ifthenelse{\boolean{LMCSstyle}}{	
	
	}{
	}

\renewcommand{\esub}[2]{[#1{\shortleftarrow}#2]}
\renewcommand{\isub}[2]{\{#1{\shortleftarrow}#2\}}

\newcommand{\resm}{\psym}
\renewcommand{\resm}{\bullet}

\newcommand{\lpos}{p}
\renewcommand{\lpos}{l}

\newcommand{\upp}{\blue{\uparrow}}
\newcommand{\downp}{\red{\downarrow}}

\newcommand{\tape}{T}

\newcommand{\pol}{d}

\newcommand{\cons}{{\cdot}}

\newcommand{\IAM}{IAM\xspace}

\newcommand{\stempty}{\epsilon}

\newcommand{\la}[1]{\lambda #1.}

\newcommand{\resmtwo}{\circ}

\newcommand\mydots{\hbox to .6em{.\hss.}}

\newcommand{\ans}{\mathsf{a}}
\newcommand{\yes}{1}
\newcommand{\no}{0}

\renewcommand{\phi}{\varphi}
\renewcommand{\epsilon}{\varepsilon}
\newcommand{\Nat}{\mathbb{N}}

\newcommand{\ttletin}[2]{\mathtt{let}\ #1\ \mathtt{in}\ #2}

\newcommand\bnfeq{\mathrel{::=}}
\newcommand\bnfalt{\mathrel{|}}

\newcommand\Tree{\mathrm{Tree}}
\newcommand\rk{\mathrm{rk}}

\newcommand\titocecilia{{\fontencoding{T5}\selectfont Nguyễn} and Pradic\xspace}

\newcommand{\downred}[1]{{\color{red}\underline{#1}}}
\newcommand{\upblue}[1]{{\color{blue}\overline{#1}}}
\newcommand{\paiam}{\mathrm{PAIAM}}
\newcommand{\apaiam}{\mathrm{APAIAM}}

\begin{document}

\title{Slightly Non-Linear Higher-Order Tree Transducers$^*$}
\titlecomment{$^*$This is the extended version of the paper presented at STACS 2025 with the same title~\cite{stacs}.}	

\author[{\fontencoding{T5}\selectfont L.~T.~D.~Nguyễn}]{{\fontencoding{T5}\selectfont Lê Thành Dũng (Tito) Nguyễn}\lmcsorcid{0000-0002-6900-5577}}[a]
\author[G.~Vanoni]{Gabriele Vanoni\lmcsorcid{0000-0001-8762-8674}}[b]

\address{Laboratoire d'Informatique et des Systèmes, CNRS \& Aix-Marseille University, France}	
\address{IRIF, Université Paris Cité, France}	





\begin{abstract}
  \noindent   We investigate the tree-to-tree functions computed by \enquote{affine
    $\lambda$-transducers}: tree automata whose memory consists of an
  affine $\lambda$-term instead of a finite state. They can be seen as
  variations on Gallot, Lemay and Salvati's Linear High-Order Deterministic Tree
  Transducers.

  When the memory is almost purely affine (\textit{à la} Kanazawa), we show that these machines can be
  translated to tree-walking transducers (and with a
  purely affine memory, we get a reversible tree-walking transducer).
  This leads to a proof of an
  inexpressivity conjecture of \titocecilia
  on \enquote{implicit automata} in an affine $\lambda$-calculus. We also prove
  that a more powerful variant, extended with preprocessing by an MSO relabeling and allowing a
  limited amount of non-linearity, is equivalent in expressive power to
  Engelfriet, Hoogeboom and Samwel's invisible pebble tree transducers.

  The key technical tool in our proofs is the Interaction Abstract Machine (IAM), an
  operational avatar of Girard's geometry of interaction---the latter is a family of semantics of
  linear logic. We work with ad-hoc specializations to $\lambda$-terms of low
  exponential depth of a tree-generating version of the IAM.
\end{abstract}

\maketitle

\section{Introduction}\label{sec:intro}

This paper investigates the expressive power of various kinds of \emph{tree
  transducers}: automata computing tree-to-tree functions. This is a topic with
a long history, and many equivalences between machine models are already known.
For instance, the class of monadic second-order transductions\footnote{\label{ftn:msot} We
  consider only tree-to-tree transductions here, but there is a rich theory of
  MSOTs between graphs, or between arbitrary relational structures,
  cf.~\cite{courcellebook}. See also~\cite{bojanblog} for a history of MSO
  transductions. In the well-studied special case of string functions, MSOTs are
  called \enquote{regular functions}, cf.~\cite{MuschollPuppis}.} (MSOTs), whose
name refers to a definition by logic, is also captured by various tree
transducer models (e.g.~\cite{EngelfrietIM21}) or by a system of primitive
functions and combinators~\cite{FOTree}. This class is closed under composition,
and includes functions such as:
\begin{align*}
  \text{mirror and add a $d$ above each $b$:} &\quad a(b(c),c) \mapsto a(c,d(b(c)))\\
  \text{relabel each $c$ by parity of its depth:} &\quad a(b(c),c) \mapsto a(b(0),1)\\
  \text{count number of non-$a$ nodes in unary:} &\quad a(b(c),c) \mapsto S(S(S(0)))
\end{align*}
Some machines for MSOTs, such as the restricted macro tree
transducers\footnote{\label{ftn:macro-alt}Which are very similar to some other
  transducer models for MSOTs: the bottom-up ranked tree transducers
  of~\cite[Sections~3.7--3.8]{BRTT} and the register tree transducers
  of~\cite[Section~4]{FOTree}.} of~\cite{MacroMSO}, involve
\emph{tree contexts} as data structures used in their computation. A tree
context is a tree with \enquote{holes} at some leaves -- for example
$b(a(c,a(\ctxhole{},c)))$ -- and these holes are meant to be substituted
by other trees. Thus, it represents a simple function taking trees to trees:
these transducers use functions as data. From this point of view, in the spirit
of functional programming, it also makes sense to consider transducers
manipulating \emph{higher-order data}, that is, functions that may take
functions as arguments (a function of order $k+1$ takes arguments of order at
most $k$). This idea goes back to the 1980s~\cite{EngelfrietHighLevel} (cf.\
Remark~\ref{rem:hltt-mtt}), and a recent variant by Gallot, Lemay and
Salvati~\cite{LambdaTransducer,gallotPhD} computes exactly the tree-to-tree
MSOTs.

\paragraph{Linear/Affine $\lambda$-Calculus for MSOTs.}

Gallot et al.~\cite{LambdaTransducer,gallotPhD} use the \emph{$\lambda$-calculus} to
represent higher-order data, and in order to control the expressive power,
they impose a \emph{linearity} restriction on $\lambda$-terms. In programming
language theory, a \emph{linear} function uses its argument \emph{exactly once},
while an \emph{affine} function must use its argument \emph{at most once} -- affineness
thus mirrors the \enquote{single-use restrictions} that appear in various tree
transducer models~\cite{AttributedMSO,MacroMSO,BRTT,FOTree}.

An independent but similar characterization of MSOTs appeared around the same
time in \titocecilia's work on
\enquote{implicit automata}~\cite{iatlc1,titoPhD}. In an analogous fashion to
how the untyped $\lambda$-calculus can be used as a Turing-complete programming
language, they consider typed $\lambda$-terms seen as standalone
programs. For a well-chosen type system -- which enforces a linearity restriction --
and input/output convention, it turns
out that the functions computed by these $\lambda$-terms are exactly the MSO tree
transductions~\cite[Theorem~1.2.3]{titoPhD}.


In fact, the proof of \cite[Theorem~1.2.3]{titoPhD} introduces, as an intermediate step, a
machine model -- the \enquote{single-state $\mathfrak{L}$-BRTTs}
of~\cite[Section~6.4]{titoPhD} -- that uses $\lambda$-terms as memory to
capture the class of MSOTs. Clearly, this device is very close to the
tree transducer model of Gallot et~al.
Yet there are also important differences between the two: they can be understood as two distinct ways of extending what we call \enquote{purely linear $\lambda$-transducers}, as we
will explain. Actually, to avoid some uninteresting pathologies (cf.~\cite[Theorem~7.0.2]{titoPhD}), we shall prefer to work with affine $\lambda$-transducers instead.
\begin{exa}\label{ex:lambda-count}
  Purely affine $\lambda$-transducers will be properly defined later, but for
  now, for the sake of concreteness, let us exhibit one such transducer.
  It takes input trees with binary $a$-labeled nodes, unary $b$-labeled nodes
  and $c$-labeled leaves, and is specified by the $\lambda$-terms:
  \[ t_a = \lambda \ell.\; \lambda r.\; \lambda x.\; \ell\; (r\; x) \qquad t_b = \lambda f.\; \lambda x.\; S\;(f\;x) \qquad t_c = S \qquad u = \lambda f.\; f\; 0 \]
  where $S$ and $0$ are constants from the output alphabet. The input
  $a(b(c),c)$ is then mapped to $u\; (t_a\;
  (t_b\;t_c)\; t_c)$ which evaluates to the normal form $S\; (S\; (S\; 0))$. So this
  $\lambda$-transducer computes the aforementioned \enquote{number of non-$a$
    nodes written in unary} function. (It actually is purely linear:
  each bound variable has exactly one occurrence.)
\end{exa}

The issue with this transducer model is its lack of expressiveness, as shown by
the following consequence of our results, which settles an equivalent\footnote{A routine
syntactic analysis akin to~\cite[Lemma~3.6]{iatlc1} shows that it is indeed
equivalent. Note that the results of~\cite{iatlc1} are indeed about affine,
not linear, $\lambda$-calculi.} conjecture
on \enquote{implicit automata} that had been put forth by \titocecilia
in~\cite[\S5.3]{iatlc1}. (We will come back in Remark~\ref{rem:additives} to how they overcome
this to characterize tree-to-tree MSO transductions.)
\begin{cor}[of Theorem~\ref{thm:main-affine} below]\label{cor:inexpress} There
  exists a regular tree\footnote{This phenomenon depends on using trees as
    inputs. Over strings, purely affine $\lambda$-terms can be used to recognize
    any regular language~\cite[Theorem~5.1]{iatlc1}.} language whose indicator
  function cannot be computed by any purely affine $\lambda$-transducer.
\end{cor}

Gallot et~al.\ avoid this limitation by extending the transducer model with common automata-theoretic
features. They then show~\cite[Chapter~7]{gallotPhD} that the usual linear $\lambda$-terms
lead to the previously mentioned characterization of MSOTs, while relaxing the
linearity condition to \enquote{almost linearity} yields the larger class of
\enquote{MSOTs with sharing}\footnote{An MSOT-S can be
  decomposed as an MSO tree-to-graph transduction (\Cref{ftn:msot}), that
  produces a rooted directed acyclic graph (DAG), followed by the unfolding of this
  graph into a tree. The rooted DAG is a compressed representation of the output tree
  with some \enquote{sharing} of subtrees.} (MSOT-S).
This notion of \emph{almost linear $\lambda$-term} was introduced by
Kanazawa~\cite{KanazawaJournal}, who also studied the almost affine case
in~\cite{KanazawaJournal,AlmostAffine}. In fact, the aforementioned
characterization of MSOT-S was first claimed by Kanazawa in a talk more than 15
years ago~\cite{KanazawaMSO}, although he never published a proof to the best of our knowledge.

\paragraph{Flavors of Affine Types.}

We shall work with an affine type system that allows some data to be marked as duplicable via the exponential modality `$!$'. The grammar of types is thus $A, B \bnfeq o \bnfalt A \multimap B \bnfalt \oc A$ -- the connective $\multimap$ is the affine function arrow.
The point is to allow us to restrict duplication by means of syntactic constraints on `!'.
\begin{defi}\label{def:almost-purely-affine}
  A \emph{purely affine} type does not contain any `!', i.e.\ is built from $o$ and $\multimap$.
  A type is said to be \emph{almost purely affine} when the only occurrences of `!' are applied to $o$. In particular, every purely affine type is almost purely affine.
\end{defi}
For example, $(o \multimap \oc o) \multimap \oc o$ is almost purely affine but not $\oc(o \multimap o)$.
The latter definition is motivated by the aforementioned almost affine
$\lambda$-terms~\cite{AlmostAffine}, which allow the
variables of base type $o$ to be used multiple times, e.g.\ $\lambda y.\;
\lambda f.\; \lambda g.\; (\lambda x.\; f\; x\; x)\; (g\;y\;y)$ is almost affine
(and even almost linear) for $x : o$ and $y : o$. Inconveniently, almost
linear/affine $\lambda$-terms are not closed under $\beta$-reduction, as
remarked in~\cite[\S4]{AlmostAffine}. For instance, the previous term reduces to
$\lambda y.\; \lambda f.\; \lambda g.\; f\; (g\;y\;y)\; (g\;y\;y)$ which uses $g
: o \multimap o \multimap o$ twice. The `$\oc$' modality provides a convenient
way to realize a similar idea while avoiding this drawback.

Each $\lambda$-transducer has in its definition a \emph{memory type}, which is for instance $o \multimap o$ for Example~\ref{ex:lambda-count}. We extend Definition~\ref{def:almost-purely-affine} to $\lambda$-transducers: a $\lambda$-transducer is purely (resp.\ almost purely) affine when its memory type is purely (resp.\ almost purely) affine.

\paragraph{Contributions.}

First, we study (almost) purely affine $\lambda$-transducers, relating them to yet another machine model: \emph{tree-walking tree transducers}, see e.g.~\cite{EngelfrietIM21}. Those are devices with a finite-state control and a reading head moving around the nodes of the input tree; in one step, the head can move to the parent or one of the children of the current node.
\begin{thm}\label{thm:main-affine}
  Writing $\subseteq$ for a comparison in expressive power, we have:
  \begin{align*}
    \text{purely affine $\lambda$-transducer} &\subseteq \text{reversible tree-walking transducer}\\
    \text{almost purely affine $\lambda$-transducer} &\subseteq \text{tree-walking transducer}
  \end{align*}
\end{thm}

Corollary~\ref{cor:inexpress} then follows immediately from a result of Bojańczyk and
Colcombet~\cite{bojanczyk2008tree}: there exists a regular tree language not
recognized by any tree-walking automaton.
We also obtain characterizations of MSOTs and of MSOT-Ss, by preprocessing the
input with an \emph{MSO relabeling} -- a special kind of MSOT that can only
change the node labels in a tree, but keeps its structure as it is.
Morally, this preprocessing amounts to the same thing as the automata-theoretic features -- finite
states and regular look-ahead -- used in~\cite{LambdaTransducer,gallotPhD}.
\begin{thm}\label{thm:msots}
  Using $\equiv$ to denote an equivalence in expressive power, we have:
  \begin{align*}
    \text{purely affine $\lambda$-transducer} \circ\text{MSO relabeling} &\equiv \text{MSOT} \\
    \text{almost purely affine $\lambda$-transducer} \circ\text{MSO relabeling} &\equiv \text{MSOT-S}
  \end{align*}
\end{thm}
This should be informally understood as a mere rephrasing
of the results of Gallot et al.\ and of Kanazawa, except with affine rather than linear $\lambda$-terms.
On a technical level, we derive our right-to-left inclusions from their results;
there turns out to be a minor mismatch, and working with affineness rather than
linearity proves convenient to overcome it. As for the left-to-right inclusions, we
get them as corollaries of Theorem~\ref{thm:main-affine}.
Finally, we give a genuinely new characterization of the class MSOT-S$^2$ of
functions that can be written as \emph{compositions of two MSO transductions with
sharing}.
\begin{defi}\label{def:almost-depth-1}
  A type is \emph{almost !-depth 1} when the only occurrences of `!' are
  applied to almost purely affine types (e.g.\ $\oc(\oc o \multimap o)$ is almost !-depth 1, but not $\oc\oc(o
  \multimap o)$).
\end{defi}
\begin{thm}\label{thm:main-depth-1}
  Almost !-depth 1 $\lambda$-transducer $\circ$ MSO relabeling $\equiv$ MSOT-S$^2$.
\end{thm}

To prove the left-to-right inclusion, we compile these $\lambda$-transducers to
\emph{invisible pebble tree transducers}~\cite{InvisiblePebbles} -- an extension
of tree-walking transducers known to capture MSOT-S$^2$. For the converse, we
rely on a simple composition procedure for $\lambda$-transducers:
\begin{prop}[proved in \S\ref{sec:lambdatrans}]\label{prop:compo}
  Suppose that $f$ and $g$ are tree-to-tree functions computed by $\lambda$-transducers with
  respective memory types $A$ and $B$. Then $g \circ f$ is computed by a
  $\lambda$-transducer with memory type $A\{o:=B\}$.
\end{prop}

\paragraph{Key Tool: the Interaction Abstract Machine (IAM)}

To translate $\lambda$-transducers into tree-walking or invisible pebble tree
transducers, we use a mechanism that evaluates a $\lambda$-term using a pointer
to its syntax tree that is moved by local steps -- just like a tree-walking
reading head. This mechanism, called the Interaction Abstract
Machine, was derived from a family of semantics of linear
$\lambda$-calculi known as \enquote{geometry of interaction} (GoI) --
for more on its history, see~\cite[Chapter~3]{vanoniPhD}. Our approach thus
differs from both the proofs of Gallot et al.\ in~\cite{gallotPhD}, which go
through a syntactic procedure for lowering the order of $\lambda$-terms, and those
of \titocecilia in~\cite{titoPhD}, which rely on another kind of denotational
semantics.

The IAM satisfies a well-known \emph{reversibility} property, see e.g.~\cite[Proposition~3.3.4]{vanoniPhD} -- it even appears in the title of the seminal paper~\cite{DanosRegnierIAM}. In the purely affine
case, this gives us reversible tree-walking transducers in
Theorem~\ref{thm:main-affine} -- these have not appeared explicitly in the
literature, but we define them similarly to the existing reversible
graph-walking automata~\cite{GraphWalking} and reversible two-way string transducers~\cite{ReversibleTransducers}. In the almost purely affine
and almost !-depth~1 cases, we use an ad-hoc optimization of the IAM that breaks
reversibility.

\paragraph{Preliminaries on Trees.} A \emph{ranked alphabet} $\Sigma$ is a finite set with a \enquote{rank} (or \enquote{arity}) $\rk(c) \in \mathbb{N} = \{0,1,\dots\}$
   associated to each \enquote{letter} $c\in\Sigma$. It can
  be seen as a first-order signature of function symbols, and by a \emph{tree}
  over $\Sigma$ we mean a closed first-order term over this signature. For instance, a tree over $\{a : 2,\; b : 1,\; c : 0\}$
  may have binary $a$-labeled nodes, unary $b$-labeled nodes and $c$-labeled leaves; examples include $a(b(c),c)$ or $b(a(a(c,c),c))$.

  We write $\Tree(\Sigma)$ for the set of trees over the ranked alphabet $\Sigma$. It will also be convenient to work with the sets $\Tree(\Sigma,X) = \Tree(\Sigma \cup \{x
  : 0 \mid x \in X\})$ of trees over $\Sigma$ with the additional possibility of having $X$-labeled leaves.

\section{Affine $\lambda$-terms and tree-to-tree $\lambda$-transducers}
\label{sec:lambda}

The grammar of our $\lambda$-terms, where $a$ ranges over constants and $x$ over variables, is
\[ t,u \bnfeq a \bnfalt x \bnfalt \lambda x.\; t \bnfalt t\; u \bnfalt \oc{t} \bnfalt \ttletin{\oc
  x = u}{t} \]
These $\lambda$-terms are considered up to renaming of
bound variables ($\lambda x.\; t$ binds
$x$ in $t$, while $\ttletin{\oc x = u}{t}$ binds $x$ in $t$ but not in~$u$).
We will also use \emph{contexts}, which are $\lambda$-terms containing one occurrence of a special symbol, the hole $\ctxhole$:
\[ \ctx,\ctxtwo,\ctxthree\bnfeq \ctxhole  \bnfalt \lambda x.\; \ctx \bnfalt t\;\ctx\bnfalt \ctx\; u \bnfalt \oc{\ctx} \bnfalt \ttletin{\oc
x = \ctx}{t} \bnfalt \ttletin{\oc x = u}{\ctx} \]
Plugging, i.e.\ substituting the hole of a context $\ctx$ for a term $\tm$, potentially capturing free variables, is written $\ctxp{\tm}$. For example,  if $C = \lambda x.\; t\; \ctxhole$, then $C\ctxholep{u\;v} = \lambda x.\; t\; (u\; v)$.
The linear logic tradition calls a term of the form $\oc{t}$ a \emph{box}, and the
number of nested boxes surrounding a subterm of some term is called its \emph{depth}
within this term. Similarly, the \emph{depth of a context} $C$ is defined to be the number $n$ of boxes surrounding the hole $\ctxhole$. We write then $C_n$. In this paper, however, we are counting only boxes which do \emph{not} surround terms of the base type $o$.

\paragraph{Typing Rules.}

Our type system is a variant of Dual Intuitionistic Linear Logic~\cite{DILL}, with
weakening to make it affine.
As already said, our grammar of types is $A, B \bnfeq o \bnfalt A \multimap B \bnfalt \oc A$.
As usual, $\multimap$ is right-associative: the parentheses in
$A\multimap(B\multimap C)$ can be dropped.

The typing contexts are of the form (unrestricted variables) | (affine
variables). In the rules below, a comma between two sets of typed affine
variables denotes a disjoint union; this corresponds to prohibiting the use of
the same affine free variable in two distinct subterms.
\[ \frac{}{\Theta \mid \Phi, x : A \vdash x : A} \qquad
  \frac{\Theta \mid \Phi, x : A \vdash t : B}{\Theta \mid \Phi \vdash \lambda
    x.\; t : A \multimap B} \qquad
  \frac{\Theta \mid \Phi \vdash t : A \multimap B \quad \Theta \mid
    \Phi' \vdash u : A}{\Theta \mid \Phi, \Phi' \vdash t\;u :
    B} \]
\[ \frac{}{\Theta, x : A \mid \Phi \vdash x : A} \qquad \frac{\Theta \mid
    \varnothing \vdash t : A}{\Theta \mid \varnothing \vdash \oc t : \oc A}
  \qquad \frac{\Theta \mid \Phi \vdash u : \oc A \quad \Theta, x : A \mid \Phi'
    \vdash t : B}{\Theta \mid \Phi, \Phi' \vdash \ttletin{\oc x = u}{t} : B}
\]
We also work with constants whose types are fixed in advance. Fixing $a : A$
means that we have the typing rule $\overline{\Theta \mid \Phi \vdash a :
  A}$. We shall also abbreviate $\varnothing \mid \varnothing \vdash t : A$ as
$t : A$.

\begin{clm}[Affineness]
  If $\lambda x.\; t$ is well-typed, the variable $x$ occurs at most once in
$t$, at depth 0. (But let-bound variables are not subject to any such restriction).
\end{clm}

\paragraph{Normalization.}

Our $\beta$-reduction rules, which can be applied in any context $\ctx$, are:
\[ L\ctxholep{\lambda x.\; t}\; u \longrightarrow_\beta L\ctxholep{t\{x:=u\}} \qquad\qquad
  \ttletin{\oc x = L\ctxholep{\oc u}}{t} \longrightarrow_\beta L\ctxholep{t\{x:=u\}}\]
where the context $L$ consists of a
succession of let-binders: $L \bnfeq \ctxhole \mid \ttletin{\oc x' = t'}{L}$.

For instance, the following is a valid $\beta$-reduction:
\[ (\ttletin{\oc x = u}{\ttletin{\oc y = v}{\lambda z.\; z\;x\;y}})\; t
  \quad\longrightarrow_\beta^*\quad \ttletin{\oc x = u}{\ttletin{\oc y = v}{t\;x\;y}} \]
This \enquote{reduction at a distance} -- an idea of
Accattoli \& Kesner~\cite{AccattoliK10} -- is a way to get the desirable Proposition~\ref{prop:normalization-simplifies}
below without having to introduce cumbersome \enquote{commuting
  conversions}. For an extended discussion in the context of a system very close
to ours, see~\cite[\S1.2.1]{MazzaHDR}.

\begin{prop}[Normalization \& subject reduction]\label{prop:norm-sr}
  Any well-typed term has a $\beta$-normal form. Furthermore, if $\Theta\mid\Phi \vdash t : A$ then $\Theta\mid\Phi \vdash t' : A$ for any $\beta$-normal form $t'$ of $t$.
\end{prop}
\begin{defi}
  We say that a term $\Theta\mid\Phi \vdash t : A$ is purely affine when \emph{all of its subterms} have purely affine
  types (cf.~Definition~\ref{def:almost-purely-affine}) and $\Theta=\emptyset$, which implies that it contains no !-box or let-binding. We also call almost purely affine (resp.\ almost !-depth 1) the terms $\Theta\mid\Phi \vdash t : A$ in which, for every subterm $u$ of $t$,
  \begin{itemize}
    \item the type of $u$ is almost purely affine (resp.\ almost !-depth 1 -- cf.~Definition~\ref{def:almost-depth-1}),
    \item if $u = \oc r$ then $r$ is purely (resp.\ almost purely) affine,
    \item $\Theta=x_1:o,\ldots,x_n:o$ (resp. $\Theta=x_1:A_1,\ldots,x_n:A_n$ where $A_1,\ldots,A_n$ are almost purely affine).
  \end{itemize}
\end{defi}

\begin{prop}\label{prop:normalization-simplifies}
  Assume that we work with purely affine constants, as will always be the
  case in this paper. Let $\Theta\mid\Phi \vdash t : A$ and \emph{suppose $t$ is in normal form}. If $A$ and the types in $\Theta$ and $\Phi$ are purely
  affine (resp.\ almost purely affine, almost !-depth 1), then so is $t$.
\end{prop}

\noindent
For example, $\ttletin{\oc z = \oc(\lambda x.\; x)}{z} : o \multimap o$ is not purely affine but its normal form $\lambda x.\; x$ is.

\paragraph{Encoding of trees in our affine $\lambda$-calculus.}

Fix a ranked alphabet $\Sigma$.
We consider $\lambda$-terms built over the constants $c : o^{\rk(c)} \multimap o$ for $c\in\Sigma$.
There is a canonical encoding $\widetilde{(\,\cdot\,)}$ of trees as closed terms of type $o$; for instance,
$\tau = a(b(c),c)$ is encoded as $\widetilde\tau = a\; (b\; c)\; c$.
\begin{prop}\label{prop:encoding}
  Every closed term \emph{of type $o$} using these constants admits a \emph{unique} normal form. Furthermore, $\widetilde{(\,\cdot\,)}$ is a \emph{bijection} between $\Tree(\Sigma)$ and these normal forms.
\end{prop}

Given a type $A$ and a family of $\lambda$-terms $\vec{t} = (t_c)_{c \in
  \Sigma}$ such that $t_c : A^{\rk(c)} \multimap A$ for each letter $c \in
\Sigma$, we write $\widehat\tau(\vec{t})$ for the result of replacing each
constant $c$ in $\widetilde\tau$ by $t_c$. It is always well typed, with type
$A$. For the example $\tau = a(b(c),c)$, we have $\widehat\tau((t_x)_{x \in \{a,b,c\}}) = t_a\; (t_b\; t_c) \;t_c$.

\paragraph{Higher-Order Transducers (or $\lambda$-Transducers).}
\label{sec:lambdatrans}

Let us fix an input alphabet $\Gamma$.
\begin{defi}
  An \emph{(affine) $\lambda$-transducer}
  $\Tree(\Gamma)\to\Tree(\Sigma)$ is specified by a \emph{memory type}~$A$ and a
  family of terms (that can use the aforementioned constants from $\Sigma$):
  \[ \underbrace{t_a : A^{\rk(a)} \multimap
      A}_{\mathclap{\text{\enquote{transition terms}}}}\ \text{for each letter}\ a \in \Gamma
    \qquad\text{and}\qquad \underbrace{u : A \multimap o}_{\mathclap{\text{\enquote{output term}}}} \]
\end{defi}
The $\lambda$-transducer defines the function
\vspace{-1em}
\[ \tau \in \Tree(\Gamma) \quad\mapsto\quad \underbrace{\sigma \in \Tree(\Sigma)\ \text{such that}\ \widetilde\sigma\ \text{is the normal form of}\ \overbrace{u\; \widehat\tau((t_a)_{a\in\Gamma})}^{\mathclap{\text{well-typed with type $o$}}}}_{\mathclap{\text{well-defined and unique thanks to Proposition~\ref{prop:encoding}}}} \]
This amounts to specifying a structurally recursive function over $\Tree(\Gamma)$ with return type $A$, followed by some post-processing that produces an output tree. Alternatively, a $\lambda$-transducer can be seen as a kind of tree automaton whose
memory consists of affine $\lambda$-terms of some type $A$ (with constants from
$\Sigma$) and whose bottom-up transitions are also defined by $\lambda$-terms.

In addition to the purely affine Example~\ref{ex:lambda-count}, we exhibit two other
$\lambda$-transducers.
\begin{exa}\label{ex:lambda-seq-nat}
  The following almost affine $\lambda$-transducer maps $S^n(0) =
  S(\dots(S(0)))$ to the list $[1,2,\dots,n]$, encoded as the tree
  $\mathtt{cons}(S(0),\dots(\mathtt{cons}(S^n(0),\mathtt{nil})\dots)$.
  \begin{align*}
    t_0 &= \lambda x.\; \mathtt{nil} : \oc o \multimap o\ \text{(memory type)}\qquad\qquad\qquad u = \lambda g.\; g\; \oc(S\; 0)\\
    t_S &= \lambda g.\; \lambda x.\; \ttletin{\oc y = x}{\mathtt{cons}\; y\; (g\; \oc(S\; y))}
  \end{align*}
\end{exa}
\begin{exa}\label{ex:lambda-bin2bin}
  The following $\lambda$-transducer takes as input the binary encoding of a
  natural number $n$ and returns a complete binary tree of height $n$, e.g.\
  $\mathtt{0}(\mathtt{0}(\mathtt{1}(\mathtt{0}(\varepsilon)))) \mapsto
  a(a(c,c),a(c,c))$. Thus, its growth is doubly exponential. Its memory type
  $\oc(\oc o \multimap \oc o) \multimap o$ is almost !-depth 1.
  \begin{align*}
    t_0 &= \lambda g.\; \lambda x.\; \ttletin{\oc f = x}{g\; \oc(\lambda y.\; f\; (f\; y))} \quad\; t_\varepsilon = \lambda x.\; \ttletin{\oc f = x}{\ttletin{\oc z = f\; \oc c}{z}}\\
    t_1 &= \lambda g.\; \lambda x.\; \ttletin{\oc f = x}{g\; \oc(\lambda y.\; \ttletin{\oc z = f\; (f\; y)}{\oc(a\; z\; z)})} \qquad u = \lambda g.\; g\; \oc(\lambda y.\; y)
  \end{align*}
\end{exa}

Moreover, composing Examples~\ref{ex:lambda-count} and~\ref{ex:lambda-seq-nat} according to Proposition~\ref{prop:compo} gives another almost purely affine example. Let us describe how this works:
\begin{proof}[Proof of Proposition~\ref{prop:compo}]
	Let $f$ and $g$ be computed by the two $\lambda$-transducers
	\[ \underbrace{(t_a : A^{\rk(a)} \multimap A)_{a\in\Gamma},\; u : A \multimap o}_{\mathclap{\text{with
				constants from}\ \Sigma}} \qquad \underbrace{(t'_c : B^{\rk(c)} \multimap B)_{c\in\Sigma},\; u' : B \multimap o}_{\mathclap{\text{with
				constants from}\ \Pi}} \]
	First, for any $v : C$ using constants $c : o^{\rk(c)} \multimap o$ for $c \in \Sigma$, an induction over the typing derivation of $v$ shows that the term $v\{c:=t'_c\}_{c \in \Sigma}$ obtained by replacing each $c$ in $v$ by $t'_c$ is well-typed, with type $C\{o:=A\}$. This means that the terms below define a $\lambda$-transducer with memory type $A\{o:=B\}$; one can check that it computes $g \circ f$ to conclude the proof:
	\[ \text{transition term (for $a\in\Gamma$):}\ t_a\{c:=t'_c\}_{c \in \Sigma} \qquad \text{output term:}\ \lambda x.\; u'\; (u\{c:=t'_c\}_{c \in \Sigma}\; x)  \qedhere \]
\end{proof}
\begin{rem}\label{rem:hltt-mtt}
  The original impetus for Engelfriet and Vogler's \enquote{high level tree transducers} was that their transducers of order $k$ are equivalent to compositions of $k$ unrestricted macro tree transducers~\cite{EngelfrietHighLevel}. A major motivation of Gallot et al.'s machine model using linear $\lambda$-terms was also function composition, for which they give efficient constructions~\cite[Chapter~6]{gallotPhD} (which are non-trivial). And in \titocecilia's \enquote{implicit automata}, composition is just a matter of plugging two $\lambda$-terms together~\cite[Lemma~2.8]{iatlc1}.
\end{rem}

\section{Tree-walking transducers (last definitions needed for Theorem~\ref{thm:main-affine})}

\paragraph{Generalities.}

In this paper, we shall encounter several machine models that generate some
output tree in a top-down fashion, starting from the root. (This is not the case
of $\lambda$-transducers.) They follow a common pattern, which we abstract as a
lightweight formalism here: essentially, a deterministic regular tree
grammar with infinitely many non-terminals.

\begin{rem}
  Engelfriet's tree grammars with storage
  (see for instance~\cite{engelfriet2014contextfree}) are more complex
  formalisms that also attempt to unify several definitions of tree transducer
  models.
\end{rem}
\begin{defi}
  A \emph{tree-generating machine} over the ranked alphabet $\Sigma$ consists of:
  \begin{itemize}
  \item a (possibly infinite) set $\mathcal{K}$ of \emph{configurations};
  \item an \emph{initial configuration} $\kappa_0 \in \mathcal{K}$ -- in concrete
    instances, $\kappa_0$ will be defined as a simple function of some input
    tree (for tree transducers) or some given $\lambda$-term (for the IAM);
  \item a \emph{computation-step (partial) function} $\mathcal{K}
    \rightharpoonup \Tree(\Sigma,\mathcal{K})$.
  \end{itemize}
\end{defi}
\begin{exa}\label{ex:machine-seq-nat}
  To motivate the formal semantics for these machines that we will soon define,
  we give a tree-generating machine that is meant to produce the list
  $[1,2,\dots,n]$ (for an arbitrarily chosen $n \in \Nat$), encoded as in
  Example~\ref{ex:lambda-seq-nat}.
  \begin{itemize}
  \item The set of configurations is $\mathcal{K} =
    \{\mathsf{spine,num}\}\times\Nat$ where $\mathsf{spine}$ and $\mathsf{num}$
    are formal symbols.
  \item The initial configuration is $(\mathsf{spine},n)$ -- let us write this
    pair as $\langle\mathsf{spine},n\rangle$.
  \item The computation-step function is
    $\left\{ \begin{aligned}
       \langle \mathsf{spine},0 \rangle &\mapsto \mathtt{nil}\\
       \langle \mathsf{spine},m+1 \rangle &\mapsto \mathtt{cons}(\langle
                                            \mathsf{num},n-m \rangle, \langle \mathsf{spine},m \rangle)\\
       \langle \mathsf{num},0 \rangle &\mapsto 0 \\
       \langle \mathsf{num},m+1 \rangle &\mapsto S(\langle \mathsf{num},m \rangle)
    \end{aligned} \right.$
  \end{itemize}
  For $n=3$, one possible run is
  \[ \langle \mathsf{spine},3 \rangle \rightsquigarrow \mathtt{cons}(\langle \mathsf{num},1
    \rangle, \langle \mathsf{spine},2  \rangle)
    \rightsquigarrow \mathtt{cons}(S(\langle \mathsf{num},0 \rangle), \langle
    \mathsf{spine},2  \rangle) \rightsquigarrow \dots \]
  All runs starting from $\langle \mathsf{spine},3 \rangle$ eventually reach the tree that encodes $[1,2,3]$.
\end{exa}

Let us now discuss the general case. Intuitively, the execution of the machine involves spawning several independent
concurrent processes, outputting disjoint subtrees. We formalize this parallel
computation as a rewriting system $\rightsquigarrow$ on
$\Tree(\Sigma,\mathcal{K})$: we have $\tau_1 \rightsquigarrow \tau_2$ whenever $\tau_2$
is obtained from $\tau_1$ by substituting one of its configuration leaves by its
image by the computation-step function. This rewriting system is orthogonal, and
therefore confluent, which means that the initial configuration has at most one
normal form. If this normal form exists and belongs to $\Tree(\Sigma)$, it is
the output of the machine; we then say that the machine converges. Otherwise, the output is undefined; the machine diverges.

\paragraph{Tree-Walking Transducers.} Before giving the definition, let us see a concrete example.

\begin{exa}\label{ex:twt-count}
  According to Theorem~\ref{thm:main-affine}, since the $\lambda$-transducer of Example~\ref{ex:lambda-count}
  is purely affine, the function \enquote{count non-$a$ nodes} that it defines
  can also be computed by some (reversible) tree-walking transducer. We show the
  run of such a transducer on the input $a_1(b_2(c_3),c_4)$ -- the
  indices are not part of the node labels, they serve to distinguish positions:
  \begin{align*}
    (q,\circlearrowleft,a_1) &\rightsquigarrow (q,\downarrow_\bullet,b_2) \rightsquigarrow S((q,\downarrow_\bullet,c_3)) \rightsquigarrow                              S(S((q,\uparrow_1^\bullet,b_2))) \rightsquigarrow S(S((q,\uparrow_1^\bullet,a_1)))\\
                             &\rightsquigarrow S(S((q,\downarrow_\bullet,c_4))) \rightsquigarrow S(S(S((q,\uparrow_2^\bullet,a_1))))) \rightsquigarrow S(S(S(0)))
  \end{align*}
  where $q$ is the single state of the transducer. The second component records
  the \enquote{provenance}, i.e.\ the previous position of the tree-walking
  transducer relatively to the current node (stored in the third component): $\downarrow_\bullet$ refers to its parent,
  $\circlearrowleft$ to itself, and $\uparrow_i^\bullet$ to its $i$-th child.
\end{exa}

\begin{defi}\label{def:twt}
  A \emph{tree-walking transducer (TWT)} $\Tree(\Gamma)\rightharpoonup\Tree(\Sigma)$ consists of:
  \begin{itemize}
  \item a finite set of \emph{states} $Q$ with an \emph{initial state} $q_0 \in Q$
  \item a family of (partial) \emph{transition functions} for $a \in \Gamma$
    \[ \delta_a\colon Q \times \{\downarrow_\bullet, \circlearrowleft, \uparrow_1^\bullet, \dots,
      \uparrow_k^\bullet\} \rightharpoonup
      \Tree(\Sigma,\; Q \times \{\uparrow_\bullet, \circlearrowleft, \downarrow^\bullet_1, \dots,
      \downarrow^\bullet_k\}) \quad\text{where}\ k = \rk(a)\]
  \item a family of (partial) \emph{transition functions at the root} for $a \in \Gamma$
  \[ \delta_a^\mathrm{root}\colon Q \times \{\circlearrowleft, \uparrow_1^\bullet, \dots,
  \uparrow_k^\bullet\} \rightharpoonup
  \Tree(\Sigma,\; Q \times \{\circlearrowleft, \downarrow^\bullet_1, \dots,
  \downarrow^\bullet_k\}) \quad\text{where}\ k = \rk(a)\]
  \end{itemize}
\end{defi}
The TWT associates to each input tree $\tau$ a tree-generating machine whose
output is the image of $\tau$. Its set of configurations is $Q \times \{\downarrow_\bullet, \circlearrowleft, \uparrow_1^\bullet, \dots\} \times \{\text{nodes of}\ \tau\}$ and its initial configuration of is $(q_0,\;\circlearrowleft,\;\text{root of \(\tau\)})$.

To define the image of $(q,p,v)$ by the computation-step function (it is
undefined if one of the following steps is undefined), we start with either
$\delta^{\mathrm{root}}_{a}(q,p)$ if $v$ is the root or $\delta_{a}(q,p)$
otherwise -- where $a$ is the label of the node $v$ -- then replace each $(q',p)\in Q
\times \{\uparrow_\bullet, \dots\}$ by
\[\begin{cases}
  (q',\;\downarrow_\bullet,\; \text{\(i\)-th child of \(v\)}) &\text{if}\ p =\; \downarrow_i^\bullet\qquad\qquad\qquad\qquad\qquad\qquad\quad
  (q',\circlearrowleft,v) \quad\text{if}\ p=\;\circlearrowleft
\\
(q',\;\uparrow_j^\bullet,\; \text{parent of \(v\)}) &\text{if}\ p=\; \uparrow_\bullet\ \text{and \(v\)
    is the \(j\)-th child of its parent}
\end{cases}\]

\begin{exa}\label{ex:twt-seq-nat}
  Using the idea of Example~\ref{ex:machine-seq-nat}, we define a tree-walking
  transducer that computes the function \enquote{$n \mapsto [1,\dots,n]$ modulo
    encodings} of Example~\ref{ex:lambda-seq-nat}. Its set of states is $Q =
  \{\mathsf{spine,num}\}$, its initial state is $\mathsf{spine}$ and its
  transitions are
  \begin{align*}
    \delta^{\mathrm{root}}_S(\mathsf{spine},\circlearrowleft) &= \delta_S(\mathsf{spine},\downarrow_\bullet) = \mathtt{cons}((\mathsf{num},\circlearrowleft),(\mathsf{spine},\downarrow^\bullet_1))\\
    \delta_S(\mathsf{num},\circlearrowleft) &= \delta_S(\mathsf{num},\uparrow_1^\bullet) = S((\mathsf{num},\uparrow_\bullet))
    \\
    \delta^{\mathrm{root}}_S(\mathsf{num},\circlearrowleft) &= \delta^{\mathrm{root}}_S(\mathsf{num},\uparrow_1^\bullet) = S(0) \qquad\qquad \delta^{\mathrm{root}}_0(\mathsf{spine},\circlearrowleft) = \delta_0(\mathsf{spine},\downarrow_\bullet) = \mathtt{nil}
  \end{align*}
\end{exa}

\begin{rem}
  In the above example, the provenance information in
  $\{\downarrow_\bullet,\dots\}$ plays no role. On the contrary, it is crucial
  in the former Example~\ref{ex:twt-count}, where we have
  $\delta_a^{\mathrm{root}}(q,\uparrow_2^\bullet) = 0$ and
  \[ \delta_a^{\mathrm{root}}(q,\circlearrowleft) =
    \delta_a(q,\downarrow_\bullet) = (q,\downarrow^\bullet_1) \qquad
    \underbrace{\delta_a^{\mathrm{root}}(q,\uparrow_1^\bullet) =
    \delta_a(q,\uparrow_1^\bullet) =
    (q,\downarrow^\bullet_2)}_{\mathclap{\text{\enquote{after returning from
          the 1st child, the traversal starts visiting the 2nd child}}}} \qquad
    \delta_a(q,\uparrow_2^\bullet) = (q,\uparrow_\bullet) \]
  In the definitions of tree-walking automata or transducers in the literature,
  e.g.~\cite{bojanczyk2008tree,Engelfriet09,EngelfrietIM21}, transitions often
  do not have access to this provenance, but instead they can depend
  on the \enquote{child number} $i$ of the current node (such that the node is an
  $i$-th child of its parent). One can easily simulate one variant with the
  other; but if neither of these features were available, the machine model would
  be strictly weaker~\cite[Theorem~5.5]{KamimuraS81}.

  Our main motivation for using provenances instead of child numbers is that,
  according to~\cite[Sections~5~and~6]{GraphWalking}, being
  \enquote{direction-determinate} -- i.e.\ knowing which previous node the
  current configuration came from -- is important in the reversible case. This
  can indeed be observed in the proof of Claim~\ref{clm:reversible}.
  (Directed states are used in~\cite{ReversibleTransducers} for a similar reason.)
\end{rem}

Let us say that a map $\delta\colon X \to \Tree(\Sigma,Y)$ is \emph{$Y$-leaf-injective} whenever in the family of trees
$(\delta(x)\mid x\in X)$, each $y \in Y$ occurs at most once: $y$ appears in at most one tree of
the family, and if it does appear, this tree has a single leaf with label $y$.
\begin{defi}
  A tree-walking transducer $\Tree(\Gamma)\rightharpoonup\Tree(\Sigma)$ is
  \emph{reversible} when all maps $\delta_a$ and $\delta^{\mathrm{root}}_a$ for
  $a\in\Gamma$ are $(Q \times \{\uparrow_\bullet, \circlearrowleft,
  \downarrow_1^\bullet, \dots,\downarrow_{\rk(a)}^\bullet\})$-leaf-injective.
\end{defi}

Example~\ref{ex:twt-count} is reversible, but not Example~\ref{ex:twt-seq-nat}.
\begin{clm}[The meaning of reversibility]
  \label{clm:reversible}
  Fix a reversible tree-walking transducer and an input tree. Any configuration
  $C$ has at most one predecessor configuration, i.e.\ one configuration $C'$
  whose image by the computation-step function contains $C$ as a leaf.
\end{clm}
\begin{proof}
	Let $(q,p,v) = C$, and assume that $(q',p',v') = C'$ exists. The provenance
	information $p$ tells us where $v'$ is situated in the input tree relatively
	to $v$, so we determine $v'$ along with its label $a$. Among the leaves of
	$\delta_a(q',p')$ (or $\delta^{\mathrm{root}}_a(q',p')$ if $v'$ is the root),
	there must be one of the form $(q,p)$ where $q$ and $p$ are related (in a
	one-to-one fashion) as defined just above Example~\ref{ex:twt-seq-nat}. By
	leaf-injectivity this uniquely determines $q'$ and $p'$.
\end{proof}

\section{The (almost) purely affine Interaction Abstract Machine}%
\label{sec:iam}

This section introduces the key technical tool for the proof of Theorem~\ref{thm:main-affine}. By
definition, computing the image of an input tree by a $\lambda$-transducer
involves normalizing a $\lambda$-term. Unfortunately, iterating
$\beta$-reductions until a normal form is reached requires too much working memory to be
implemented by a finite-state device, such as a tree-walking transducer. That is
why we rely on other ways to normalize terms of base type $o$, namely variants
of the Interaction Abstract Machine (IAM)~\cite{mackie_geometry_1995,DanosRegnierIAM,AccattoliLV20}.

\paragraph{The Purely Affine IAM}
Let us start with a machine that can normalize a purely affine term~$v : o$.
Intuitively, it moves a token around the \emph{edges} of the syntax tree of $v$; we represent the
situation where the token is on the edge connecting the subterm $t$ to its
context $C$ (such that $C\ctxholep{t}=v$) and is moving {\color{red}down}
(resp.\ {\color{blue}up}) as $C\ctxholep{\downred{t}}$ (resp.\
$C\ctxholep{\upblue{t}}$). The token carries a small amount of additional
information: a \enquote{tape} which is a stack using the symbols $\resm$ and
$\resmtwo$.

We reuse the formalism of tree-generating machines from the previous section, and we denote by $X^*$ (with the Kleene star) the set of lists with elements in $X$.
\begin{defi}\label{def:paiam}
  Let $v : o$ be purely affine. The tree-generating machine $\paiam(v)$ has:
  \begin{itemize}
  \item configurations of the form $(d,t,C,T)$ where $d \in \{\downp,\upp\}$,
    $C\ctxholep{t} = v$ and $T \in \{\resm,\resmtwo\}^*$ -- which we abbreviate
    as $(C\ctxholep{\downred{t}},T)$ when $d=\downp$ or
    $(C\ctxholep{\upblue{t}},T)$ when $d=\upp$;
  \item the initial configuration $(\downred{t},\varepsilon) =
    (\downp,t,\ctxhole,\varepsilon)$;
  \item the following computation-step function:
    \[
\begin{array}{rclcrcl}
  (C\ctxholep{\downred{t\,u}}, T) & \mapsto & (C\ctxholep{\downred{t}\,u}, \resm\cons T) & \qquad& (C\ctxholep{\upblue{t}\,u}, \resm\cons T)& \mapsto& (C\ctxholep{\upblue{t\,u}}, T)\\
  (C\ctxholep{t\,\upblue{u}}, T) & \mapsto & (C\ctxholep{\downred{t}\,u}, \resmtwo\cons T) & \qquad& (C\ctxholep{\upblue{t}\,u}, \resmtwo\cons T)& \mapsto& (C\ctxholep{t\,\downred{u}}, T)\\
  (C\ctxholep{\lambda x.\, \upblue{t}}, T) & \mapsto & (C\ctxholep{\upblue{\lambda x.\, t}}, \resm\cons T) & \qquad & (C\ctxholep{\downred{\lambda x.\, t}}, \resm\cons T) & \mapsto & (C\ctxholep{\lambda x.\, \downred{t}}, T) \\
  (C\ctxholep{\lambda x.\, D\ctxholep{\downred{x}}}, T) & \mapsto & (C\ctxholep{\upblue{\lambda x.\, D\ctxholep{x}}}, \resmtwo\cons T) &\qquad&
   (C\ctxholep{\downred{\lambda x.\, D\ctxholep{x}}}, \resmtwo\cons T) &\mapsto&
   (C\ctxholep{\lambda x.\, D\ctxholep{\upblue{x}}}, T)\\
  (C\ctxholep{\downred{c}},\resm^{\rk(c)}\cons T) & \mapsto & \multicolumn{5}{l}{c\bigl( (C\ctxholep{\upblue{c}},\resmtwo\cons\tape),(C\ctxholep{\upblue{c}},\resm\cons\resmtwo\cons\tape),\dots,(C\ctxholep{\upblue{c}},\resm^{\rk(c)-1}\cons\resmtwo\cons\tape) \bigr)}
\end{array}
\]
  \end{itemize}
\end{defi}
The last rule handles constants $c : o^{\rk(c)} \multimap o$ coming from a
fixed ranked alphabet~$\Sigma$. When $\rk(c)=0$, the right-hand side is simply
the tree with a single node labeled by $c$; otherwise, it is a tree of height 1,
whose leaves are all configurations.
\begin{exa}\label{ex:iam-count}
  Let $u$ and $(t_x)_{x\in\{a,b,c\}}$ be the terms defining the purely affine
  $\lambda$-transducer from Example~\ref{ex:lambda-count}. On the term $v = u\; (t_a\;
  (t_b\;t_c)\; t_c)$, we have the following IAM execution depicted in \Cref{fig:ex}.

 \begin{figure}
 	\begin{align*}
 		(\downred{v},\stempty)
 		&\rightsquigarrow^{\phantom{*}}
 		(\downred{u}\; (t_a\; (t_b\;t_c)\; t_c),\resm)
 		\rightsquigarrow
 		((\lambda f.\; \downred{f\; 0})\; (t_a\; (t_b\;t_c)\; t_c),\stempty)
 		\rightsquigarrow
 		((\lambda f.\; \downred{f}\; 0)\; (\dots),\resm)\\
 		&\rightsquigarrow^{\phantom{*}}
 		(\upblue{(\lambda f.\; f\; 0)}\; (t_a\; (t_b\;t_c)\; t_c),\resmtwo\resm)
 		\rightsquigarrow
 		(u\; \downred{(t_a\; (t_b\;t_c)\; t_c)},\resm)
 		\rightsquigarrow^2
 		(u\; (\downred{t_a}\; (t_b\;t_c)\; t_c),\resm\resm\resm)\\
 		&\rightsquigarrow^4
 		(u\; ((\lambda \ell.\; \lambda r.\; \lambda x.\; \downred{\ell}\; (r\; x))\; (t_b\;t_c)\; t_c),\resm)
 		\rightsquigarrow
 		(u\; (\upblue{t_a}\; (t_b\;t_c)\; t_c),\resmtwo\resm)
 		\rightsquigarrow^2
 		(\dots\downred{t_b}\dots,\resm\resm)\\
 		&\rightsquigarrow^3
 		(u\; (t_a\; ((\lambda f.\; \lambda x.\; \downred{S}\;(f\;x))\;t_c)\; t_c),\resm) \rightsquigarrow
 		\underbracket{S}_{\mathclap{\qquad\qquad\qquad\qquad\text{output node (will be the root of the output tree)}}}((u\; (t_a\; ((\lambda f.\; \lambda x.\; \upblue{S}\;(f\;x))\;t_c)\; t_c),\resmtwo))\\[-4mm]
 		&\rightsquigarrow^* \ldots [\text{many steps}] \ldots\\
 		&\rightsquigarrow^*
 		S(S((u\;(\dots\downred{t_c}),\resm)))
 		\rightsquigarrow
 		S^3((u\; (t_a\; (t_b\;t_c)\; \upblue{S}),\resmtwo))
 		\rightsquigarrow
 		S^3((u\; (\downred{t_a\; (t_b\;t_c)}\; S),\resmtwo\resmtwo))
 		\\
 		&\rightsquigarrow^2 S^3((u\; ((\lambda \ell.\; \downred{\lambda r.\; \lambda x.\; \ell\; (r\; x)})\; (t_b\;t_c)\; t_c),\resmtwo\resmtwo))
 		\rightsquigarrow S^3((\dots(\upblue{r}\; x)\dots,\resmtwo))\\
 		&\rightsquigarrow^{\phantom{*}} S^3((\dots(r\; \downred{x})\dots,\stempty))
 		\rightsquigarrow S^3((u\; ((\lambda \ell.\; \lambda r.\; \upblue{\lambda x.\; \ell\; (r\; x)})\; (t_b\;t_c)\; t_c),\resmtwo))\\
 		&\rightsquigarrow^2
 		S^3((u\; (\upblue{t_a}\; (t_b\;t_c)\; t_c),\resm\resm\resmtwo))
 		\rightsquigarrow^2
 		S^3((u\; \upblue{(t_a\; (t_b\;t_c)\; t_c)},\resmtwo))
 		\rightsquigarrow
 		S^3((\downred{u}\dots,\resmtwo\resmtwo))
 		\\
 		&\rightsquigarrow^{\phantom{*}}
 		S^3(((\lambda f.\; \upblue{f}\; 0)\dots,\resmtwo))
 		\rightsquigarrow
 		S^3(((\lambda f.\; f\; \downred{0})\dots,\stempty))
 		\rightsquigarrow S(S(S(0)))
 	\end{align*}
 	\caption{Execution trace of the term $v = u\; (t_a\;
 		(t_b\;t_c)\; t_c)$, where $u$ and $(t_x)_{x\in\{a,b,c\}}$ are the terms defining the purely affine
 		$\lambda$-transducer from Example~\ref{ex:lambda-count}. $C \rightsquigarrow^n C'$ means that $C$ rewrites into $C'$ in $n$ steps. }
 	\label{fig:ex}
 \end{figure} 
\end{exa}

Aside from our bespoke extension dedicated to constants from $\Sigma$, all the
other rules in Definition~\ref{def:paiam} come from the \enquote{Linear IAM} described by
Accattoli et al.~\cite[Section~3]{AccattoliLV20} (see
also~\cite[\S3.1]{vanoniPhD}) -- we refer to those papers for high-level
explanations of these rules. Despite its name, the Linear IAM also works for
affine terms (cf.~\cite[\S3.3.4]{MazzaHDR}).
In particular, when run on closed normal forms of type $o$, namely encoding of trees by Prop.~\ref{prop:encoding}, the IAM outputs the encoded tree.
\begin{prop}\label{prop:paiam-on-normal-form}
	For each tree $\tau\in \Tree(\Sigma)$,
	$(\red{\underline{\widetilde{\tau}}},\stempty)\rightsquigarrow^* \tau$.
\end{prop}
\begin{proof}
	We strengthen the statement, considering the reduction $(\ctxp{\red{\underline{\widetilde{\tau}}}},\stempty)\rightsquigarrow^* \tau$. Then, we proceed by induction on the structure of $\tau$. If $\tau$ is a leaf, then we have $\widetilde{\tau}=c$ and $\rk(c)=0$. Thus, we have
	$  ( \ctxp{\red{\underline{c}}},  \stempty) \rightsquigarrow
	c$, which concludes this case. Otherwise $\tau$ is not a leaf, and in particular it is of the form $c\; (\tau_1\ldots \tau_k)$, where $k=\rk(c)\geq 1$. Then we have:
	\[
	\begin{array}{rll}(\ctxp{\red{\underline{c\; \widetilde{\tau}_1\ldots \widetilde{\tau}_k}}},\stempty) &\rightsquigarrow^k & (\ctxp{ \red{\underline{c}}\; \widetilde{\tau}_1\ldots \widetilde{\tau}_k},\resm^k)\\
		&\rightsquigarrow&c\;\bigl(  (\ctxp{\blue{\overline{c}}\; \widetilde{\tau}_1\ldots \widetilde{\tau}_k}), \resmtwo,
		\dots,
		( \ctxp{\blue{\overline{c}}\; \widetilde{\tau}_1\ldots \widetilde{\tau}_k}  ,\resm^{k-1}\cons\resmtwo)\bigr)\\
		&\rightsquigarrow^*& c\;\bigl((\ctxp{c\; \red{\underline{\widetilde{\tau}_1}} \widetilde{\tau}_2\ldots \widetilde{\tau}_k}, \stempty),
		\dots,
		( \ctxp{c\;\widetilde{\tau}_1\ldots \widetilde{\tau}_{k-1}\red{\underline{\widetilde{\tau}_k}}} , \stempty)\bigr)\\
		&\rightsquigarrow^*& c\; (\tau_1\ldots \tau_k)
	\end{array}
	\]
	(the last reduction comes from applying the induction hypothesis to each configuration).
\end{proof}

\begin{thm}[Soundness of the purely affine IAM]\label{thm:paiam-soundness}
  For any purely affine term $v : o$, the output of $\paiam(v)$ is the unique
  $\tau\in\Tree(\Sigma)$ such that $\widetilde\tau$ is the normal form of $v$.
\end{thm}

This comes from fact that what is computed by the IAM, in our case the tree, is invariant by $\beta$-reduction. Since this is a minor variant of standard results (see e.g.~\cite{AccattoliLV20,Parsimonious}), we relegate the proof to \Cref{sec:paiam-soundness}.


\paragraph{Bounding Tapes via a Typing Invariant.}

The size of tapes $T\in\{\resm,\resmtwo\}^*$ can be controlled by leveraging the type
system. The idea is that a tape that appears in an IAM run \enquote{points to}
an occurrence of the base type $o$ in the type of the current subterm.
Formally, we define inductively:
\[ A \lightning \varepsilon = A \qquad (A \multimap B) \lightning (\resmtwo\cons T)
  = A\lightning T \qquad (A \multimap B) \lightning (\resm\cons T) = B\lightning T  \]
(thus, $o \lightning T$ is undefined for $T \neq \varepsilon$). We then have the
following invariant on configurations:
\begin{prop}[compare with~{\cite[Lemma~32 in \S3.3.5]{MazzaHDR}}]\label{prop:typing-invariant}
  Suppose that either $(C\ctxholep{\downred{t}},T)$ or
  $(C\ctxholep{\upblue{t}},T)$ appears in a run of $\paiam(v)$ for some $v:o$. If the
  (not necessarily closed) term $t$ is given the type $A$ as part of a typing
  derivation for $v : o$, then $A \lightning T = o$.
\end{prop}
Since $|T| \leq \mathrm{height}(A)$ is a necessary condition for $A \lightning
T$ to be defined, we directly get:
\begin{cor}\label{cor:bound-height}
  The sizes of the tapes that appear in a run of $\paiam(v)$ are bounded by the
  maximum, over all subterms $t$ of $v$, of the height of the syntax tree of the
  type of $v$.
\end{cor}

\paragraph{The Almost Purely Affine Case.}

To handle the second half of Theorem~\ref{thm:main-affine}, we add non-standard rules
for let-bindings and !-boxes to the Interaction Abstract Machine.
\begin{defi}\label{def:apaiam}
  Let $v : o$ be an almost purely affine term. $\apaiam(v)$ is the extension of
  $\paiam(v)$ (cf.\ Definition~\ref{def:paiam}) with the following new cases in the
  computation-step function:
  \begin{align*}
    (C\ctxholep{\downred{\ttletin{\oc x = u}{t}}},T) &\mapsto (C\ctxholep{\ttletin{\oc x = u}{\downred{t}}},T) \qquad\quad (C\ctxholep{\downred{\oc t}},T) \mapsto (C\ctxholep{\oc\downred{t}},T)\\
    (C\ctxholep{\ttletin{\oc x = u}{\upblue{t}}},T) &\mapsto (C\ctxholep{\upblue{\ttletin{\oc x = u}{t}}},T) \\
    (C\ctxholep{\ttletin{\oc x = u}{D\ctxholep{\downred{x}}}},T) &\mapsto (C\ctxholep{\ttletin{\oc x = \downred{u}}{D\ctxholep{x}}},T)
  \end{align*}
\end{defi}
\begin{clm}\label{clm:typing-invariant}
  Proposition~\ref{prop:typing-invariant} extends to $\apaiam(v)$ with $\oc A \lightning T
  = A \lightning T$.
\end{clm}
The last rule (and the rule for !-boxes) in Definition~\ref{def:apaiam} break(s) a key duality principle at work in the purely affine IAM
(and suggested by the layout of Definition~\ref{def:paiam}):
if any rule -- except the one for constants from $\Sigma$ -- sends a
configuration $\kappa_1$ to another configuration~$\kappa_2$, then there is a
dual rule sending $\kappa_2^\perp$ to~$\kappa_1^\perp$, where
$(C\ctxholep{\downred{t}},T)^\perp=(C\ctxholep{\upblue{t}},T)$ and conversely
$(C\ctxholep{\upblue{t}},T)^\perp=(C\ctxholep{\downred{t}},T)$.

In fact, our new rule for let-bound variables $\downred{x}$ \emph{cannot} have a dual, because it is not injective.
Indeed, consider a term of the form $C\ctxholep{\ttletin{\oc x = u}{t}}$ where $t$ contains multiple occurrences of $x$, i.e.\ $t = D_1\ctxholep{x} = D_2\ctxholep{x}$ for some contexts $D_1 \neq D_2$ -- this may happen, since let-bound variables are not affine. Then for any $T\in\{\resm,\resmtwo\}^*$, the computation-step function sends both $(C\ctxholep{\ttletin{\oc x = u}{D_1\ctxholep{\downred{x}}}},T)$ and $(C\ctxholep{\ttletin{\oc x = u}{D_2\ctxholep{\downred{x}}}},T)$ to the same configuration $(C\ctxholep{\ttletin{\oc x = \downred{u}}{t}},T)$ due to the $\downred{x}$ rule. This is why \emph{reversible} TWTs can simulate the purely affine IAM, but not the \emph{almost} purely affine IAM.

To be sure that the missing dual rule is unnecessary, we show that
configurations of the form $(C\ctxholep{\ttletin{\oc x = \upblue{u}}{t}},T)$
cannot occur in an actual run. We need an invariant for reachable configurations first.
\begin{prop}\label{prop:dir}
	If $(C\ctxholep{\upblue{t}},T)$ (resp.\ $(C\ctxholep{\downred{t}},T)$) appears in a run of $\apaiam(v)$ for some almost purely affine $v : o$, then $T$ contains an odd (resp.\ even) number of $\resmtwo$.
\end{prop}
Now we can show the following:
\begin{prop}
	For any almost purely affine $v : o$:
	\begin{enumerate}
		\item $(C\ctxholep{\ttletin{\oc x = \upblue{u}}{t}},T)$ cannot appear in a run of $\apaiam(v)$;
		\item $(C\ctxholep{\oc\upblue{t}},T)$ cannot appear in a run of $\apaiam(v)$;
	\end{enumerate} 
\end{prop}
\begin{proof}
	We prove the first point, the second one is equivalent. Assume the opposite for the sake of
	contradiction. The typing rule for let-bindings forces the type of $u$ to have
	the form $\oc A$, and by almost pure affineness, it must be $\oc o$.
	By Claim~\ref{clm:typing-invariant}, $\oc o \lightning T = o\lightning T$; therefore, $T =
	\varepsilon$, which is not of odd length and thus contradicts  Prop.~\ref{prop:dir}.
\end{proof}
Having ruled out these problematic configurations, we can establish soundness for the almost purely affine IAM exactly as before, extending Theorem~\ref{thm:paiam-soundness}.
\begin{prop}[Soundness of the almost purely affine IAM, proved in \Cref{sec:paiam-soundness}]\label{prop:apaiam-soundness}
	For any almost purely affine term $v : o$, the output of $\apaiam(v)$ is the unique
	$\tau\in\Tree(\Sigma)$ such that $\widetilde\tau$ is the normal form of $v$.
\end{prop}

\section{Simulating $\lambda$-transducers by tree-walking: proof of Theorem~\ref{thm:main-affine}}

Fix an almost purely affine $\lambda$-transducer given by $u : A \multimap o$ and $t_a : A^{\rk(a)} \multimap A$ for $a\in\Gamma$. Thanks to Proposition~\ref{prop:norm-sr}, we may
assume that $u$ and every $t_a$ are in normal form.
Proposition~\ref{prop:normalization-simplifies} then tells us that since $A$ is almost purely
affine, so are these terms -- and therefore, so is $u\,
\widehat\tau((t_a)_{a\in\Gamma}) : o$ for any input $\tau\in\Tree(\Gamma)$.
Thus, we can use the almost purely affine IAM to compute the tree encoded by its normal
form -- which, by definition, gives us the image of $\tau$ by our
$\lambda$-transducer.

To prove Theorem~\ref{thm:main-affine}, we simulate $\apaiam(u\, \widehat\tau((t_a)_{a\in\Gamma}))$ by a tree-walking transducer running on $\tau$. The states of this TWT are defined as follows. Let $H$ be a bound on the length of tapes that can appear in a run of $\apaiam(u\, \widehat{\tau}((t_a)_{a\in\Gamma}))$, independent of $\tau\in\Tree(\Gamma)$ -- we know such a bound exists thanks to Corollary~\ref{cor:bound-height} and Claim~\ref{clm:typing-invariant}. For a formal symbol $\mathsf{X} \in \{\mathsf{I,U,T},\nabla,\Delta\}$, we write $\mathsf{X}(x,y,z)$ for the tuple $(x,y,z)$ tagged with the constructor name $\mathsf{X}$ (think of algebraic data types in functional programming). We also use the formal symbols $\lozenge_i$ as constants in $\lambda$-terms. 
\begin{defi}
	A state is one of the following, for $d\in\{\upp,\downp\}$ and $T\in\{\resm,\resmtwo\}^{\leqslant H}$:
	\begin{itemize}
		\item the initial state $\mathsf{I}$
		\item $\mathsf{T}(d,t,C,T)$ such that $\exists a \in \Gamma : C\ctxholep{t} = t_a\; \lozenge_1\; \dots \; \lozenge_{\rk(a)}$ and $C \neq \ctxhole$ and $\forall i,\, t \neq \lozenge_i$
		\item $\mathsf{U}(d,t,C,T)$ such that $C\ctxholep{t} = u$
		\item $\nabla(T)$ or $\Delta(T)$
	\end{itemize}
\end{defi}
Note that using $\{\resm,\resmtwo\}^{\leqslant H}$ rather than $\{\resm,\resmtwo\}^*$ makes the set of states \emph{finite}.
Before further detailing the construction of the TWT, let us illustrate its execution on an example. We reuse the \enquote{colored term in context} notation from the previous section for the tuples $(d,t,C,T)$.
\begin{exa}\label{ex:simul-count}
  We translate Example~\ref{ex:lambda-count} to a TWT that has the following run on the input $a_1(b_2(c_3),c_4)$ -- note how the steps correspond to those of the IAM run in Example~\ref{ex:iam-count}.
  This run visits the same input nodes as Example~\ref{ex:twt-count}, in the same order. The only difference is that it stays for longer on each node ($\circlearrowleft$ appears very frequently).
  \begin{align*}
    (\mathsf{I},\circlearrowleft,a_1)
    &\rightsquigarrow^{\phantom{*}}
    (\mathsf{U}(\downred{\lambda f.\; f\; 0},\resm),\circlearrowleft,a_1)
  \rightsquigarrow^3
    (\mathsf{U}(\upblue{\lambda f.\; f\; 0},\resmtwo\resm),\circlearrowleft,a_1)
  \rightsquigarrow
    (\nabla(\resm),\circlearrowleft,a_1)\\
  &\rightsquigarrow^{\phantom{*}}
    (\mathsf{T}(\downred{t_a\;\lozenge_1}\;\lozenge_2,\resm\resm),\circlearrowleft,a_1)
  \rightsquigarrow^5
    (\mathsf{T}((\lambda \ell.\; \lambda r.\; \lambda x.\; \downred{\ell}\; (r\; x))\;\lozenge_1 \;\lozenge_2,\resm),\circlearrowleft,a_1)\\
  &\rightsquigarrow^{\phantom{*}}
    (\mathsf{T}(\upblue{t_a}\;\lozenge_1\;\lozenge_2,\resmtwo\resm),\circlearrowleft,a_1)
  \rightsquigarrow
    (\nabla(\resm),\downarrow_\bullet,b_2)
    \rightsquigarrow
    (\mathsf{T}(\downred{t_b}\;\lozenge_1,\resm\resm),\circlearrowleft,b_2)\\
  &\rightsquigarrow^3
    (\mathsf{T}((\lambda f.\; \lambda x.\; \downred{S}\;(f\;x))\;\lozenge_1,\resm),\circlearrowleft,b_2) \rightsquigarrow
    S((\mathsf{T}((\lambda f.\; \lambda x.\; \upblue{S}\;(f\;x))\;\lozenge_1,\resmtwo),\circlearrowleft,b_2))\\
  &\rightsquigarrow^* \ldots [\text{many steps}] \ldots\\
  &\rightsquigarrow^* S(S((\nabla(\resm),\downarrow_\bullet,c_4)))
  \rightsquigarrow
  S^3((\Delta(\resmtwo),\uparrow^\bullet_2,a_1))
  \rightsquigarrow
  (\mathsf{T}(\downred{t_a\;\lozenge_1}\;\lozenge_2,\resmtwo\resmtwo),\circlearrowleft,a_1)
  \\
  &\rightsquigarrow^7
  S^3(\mathsf{T}(\upblue{t_a}\; \lozenge_1\; \lozenge_2),\resm\resm\resmtwo),\circlearrowleft,a_1)
  \rightsquigarrow^2
  S^3((\Delta(\resmtwo),\circlearrowleft,a_1))
  \rightsquigarrow
  S^3(\mathsf{U}(\downred{u},\resmtwo\resmtwo),\circlearrowleft,a_1)
  \\
  &\rightsquigarrow^2
    S^3((\mathsf{U}(\lambda f.\; f\; \downred{0},\stempty),\circlearrowleft,a_1))
    \rightsquigarrow S(S(S(0)))
\end{align*}
\end{exa}

\noindent
Turning back to the general case, let us \enquote{translate} IAM configurations reachable from $(\downp \mid u\; \widehat\tau((t_a)_{a\in\Gamma}) \mid \ctxhole \mid \varepsilon)$ into TWT configurations on the input tree $\tau$ -- consisting of a state, a \enquote{provenance} in $\{\uparrow_\bullet,\circlearrowleft,\downarrow_1^\bullet,\dots\}$ and a node in $\tau$. We have not yet defined the transition functions of the TWT, but they are not involved in defining its configurations; they will later be chosen to ensure that this translation is a simulation (Lemma~\ref{lem:iam-twt-simul}).

Given a reachable APAIAM configuration, i.e.\ a tuple $\mathtt{iamC} = (d,t,C,T)$ where $d\in\{\upp,\downp\}$, $T\in\{\resm,\resmtwo\}^{\leqslant H}$ and $C\ctxholep{t} = u\; \widehat\tau$, we define $\mathsf{sim}(\mathtt{iamC})$ by a case analysis on where the occurrence of $t$ as a subterm of $u\; \widehat\tau$ pinpointed by $C$ is:
\begin{itemize}
  \item $t$ is a subterm of $u$ and $C = D\; \widehat\tau$ for some $D$ (which implies $u = D\ctxholep{t}$). In this case, we take $\mathsf{sim}(\mathtt{iamC}) = (\mathsf{U}(d,t,D,T),\; \circlearrowleft,\; \text{root node})$.
  \item $t$ occurs inside $\widehat\tau$. In that case, let $s = \widehat{\tau'}$ be the smallest subterm of this form in $\widehat\tau$ that contains $t$ as a subterm, let $\nu$ be the node of $\tau$ where the subtree $\tau'$ is rooted, and let $a$ be the label of $\nu$. We have $s = t_a\, \widehat{\tau_1}\, \dots\, \widehat{\tau_{\rk(a)}}$ where the $\tau_i$ are the subtrees rooted at the children of $\nu$. Since $t$ is not a subterm of any $\widehat{\tau_i}$ by minimality of $s$, we have $t = r\{\lozenge_i := \widehat{\tau_i}\; \forall i\}$ where $r$ is a subterm of $t_a\, \lozenge_1\, \dots\, \lozenge_{\rk(a)}$ that is not equal to any $\lozenge_i$. Let us write $t_a\, \lozenge_1\, \dots\, \lozenge_{\rk(a)} = D\ctxholep{r}$.
    \begin{itemize}
      \item If $r$ is a strict subterm, i.e.\ $D \neq \ctxhole$, then we take $\mathsf{sim}(\mathtt{iamC}) = (d,r,D,T)$.
      \item If $D = \ctxhole$ and $d = \downp$, i.e.\ $\mathtt{iamC} = (C\ctxholep{\downred{s}},T)$, then we take $\mathsf{sim}(\mathtt{iamC}) = (\nabla(T),p,\nu)$ where $p=\circlearrowleft$ (resp.\ $p =\downarrow_\bullet$) when $\nu$ is (resp.\ is not) the root.
      \item Finally, if $D = \ctxhole$ and $d = \upp$, i.e.\ $\mathtt{iamC} = (C\ctxholep{\upblue{s}},T)$, then:
            \begin{itemize}
              \item if $\nu$ is the $i$-th child of some node $\xi$, then we take $\mathsf{sim}(\mathtt{iamC}) = (\Delta(T),\uparrow_i^\bullet,\xi)$;
              \item otherwise (when $\nu$ is the root) we take $\mathsf{sim}(\mathtt{iamC}) = (\Delta(T),\circlearrowleft,\nu)$.
            \end{itemize}
    \end{itemize}
	\item The only remaining possibility is that $t = u\;\widehat\tau$ and $C = \ctxhole$. We know since $u\;\widehat\tau : o$ and thanks to the typing invariants (Propositions~\ref{prop:typing-invariant} and~\ref{prop:dir}) that $d = \downp$. Therefore this is the initial configuration of the IAM: we take $\mathsf{sim}(\mathtt{iamC}) = \mathsf{I}$.
\end{itemize}

\begin{lem}\label{lem:iam-twt-simul}
  There exists a tree-walking transducer with set of states $Q$ such that for every input, the above map $\mathsf{sim}$ is a step-by-step simulation, i.e.\ the following diagram commutes:
\[\begin{tikzcd}
	{\text{IAM config}} &&& {\Tree(\Sigma,\; \text{IAM config})} \\
	\\
	{\text{TWT config}} &&& {\Tree(\Sigma,\; \text{TWT config})}
	\arrow["{\text{computation step}}", from=1-1, to=1-4]
	\arrow["{\mathsf{sim}}"', from=1-1, to=3-1]
	\arrow["{\text{apply}\ \mathsf{sim}\ \text{to config leaves}}", from=1-4, to=3-4]
	\arrow["{\text{computation step}}"', from=3-1, to=3-4]
\end{tikzcd}\]  
	Furthermore, when the memory type $A$ of our $\lambda$-transducer is purely affine, we can guarantee that the tree-walking transducer that we obtain is reversible.
\end{lem}

\begin{proof}[Proof sketch]
	Let $a\in\Gamma$, $q\in Q$ and $p\in\{\uparrow_\bullet,\circlearrowleft,\downarrow_1,\dots,\downarrow_{\rk(a)}^\bullet\}$. We must choose $\delta_a(q,p)$ and $\delta^{\mathrm{root}}_a(q,p)$ so that for any input tree $\tau$, the simulation condition is verified for all reachable IAM configurations sent by the translation to $(q,p,\nu)$ for some node $\nu$ of $\tau$ with label $a$. (If no input tree gives rise to such an IAM configuration, we leave both $\delta_a(q,p)$ and $\delta^{\mathrm{root}}_a(q,p)$ undefined.)
	This can be done by performing a case analysis on $q$ and $p$ while looking at the definition of the computation-step function for the IAM. Morally, this works because the rules of the IAM act \enquote{locally}.

	Let us work out a case chosen to illustrate several subtleties involved:
	$q = \mathsf{T}(D'\ctxholep{\downred{x}},T)$ where $x$ is a variable of the $\lambda$-calculus, and $p = \circlearrowleft$. Necessarily, $D' = D\; \lozenge_1\; \dots\; \lozenge_{\rk(a)}$. A corresponding IAM configuration must be of the form $(C\ctxholep{D\ctxholep{\downred{x}}},T)$ where $C$ is built from $\ctxhole$, $u$ and the $t_b$ ($b\in\Gamma$) by using only applications. Therefore, the hole $\ctxhole$ in $C$ is not in the scope of any binder, which means that the binder of $x$ in $C\ctxholep{D\ctxholep{x}} = u\; \widehat{\tau}$ must be inside $D$. There are two subcases:
	\begin{itemize}
		\item $x$ is $\lambda$-bound: $D = D_1\ctxholep{\lambda x. D_2}$. Then, by definition, the computation-step function sends the IAM configuration $(C\ctxholep{D\ctxholep{\downred{x}}},T)$ to $(C\ctxholep{D_1\ctxholep{\upblue{\lambda x. D_2\ctxholep{x}}}}, \resmtwo\cdot T)$. This almost tells us how we should define $\delta_a(q,\circlearrowleft)$ and $\delta^{\mathrm{root}}_a(q,\circlearrowleft)$, but we need another case analysis:
		\begin{itemize}
			\item If $D_1 = \ctxhole$ and $\rk(a)=0$, which means that we would be exiting the subterm that corresponds to a leaf of the input tree, we must take $\delta_a(q,\circlearrowleft) = (\Delta(\resmtwo\cdot T),\uparrow_\bullet)$ and $\delta^{\mathrm{root}}_a(q,\circlearrowleft) = (\Delta(\resmtwo\cdot T),\circlearrowleft)$.
			\item Otherwise, $\delta_a(q,\circlearrowleft) = \delta^{\mathrm{root}}_a(q,\circlearrowleft) = (\mathsf{T}(D_1\ctxholep{\upblue{\lambda x. D_2\ctxholep{x}}} \,\lozenge_1\, \dots\, \lozenge_{\rk(a)},\; \resmtwo\cdot T),\; \circlearrowleft)$ is the choice that makes the simulation work.
		\end{itemize}
		Note that in all these sub-sub-cases, to define a valid state, $\resmtwo\cdot T$ must belong to $\{\resm,\resmtwo\}^{\leqslant H}$. This will be the case for reachable IAM configurations, so if the token $T$ is too long, we can simply leave $\delta_a(q,\circlearrowleft)$ and $\delta^{\mathrm{root}}_a(q,\circlearrowleft)$ undefined.
		\item $x$ is $\mathtt{let}$-bound: $D = D_1\ctxholep{\ttletin{\oc x = v}{D_2}}$. The computation-step function sends the IAM configuration to $(C\ctxholep{D_1\ctxholep{\ttletin{\oc x = \downred{v}}{D_2\ctxholep{x}}}},\;T)$. We therefore define both $\delta_a(q,\circlearrowleft)$ and $\delta^{\mathrm{root}}_a(q,\circlearrowleft)$ to be $(\mathsf{T}(D_1\ctxholep{\ttletin{\oc x = \downred{v}}{D_2\ctxholep{x}}}\,\lozenge_1\, \dots\, \lozenge_{\rk(a)},\;T),\;\circlearrowleft)$.
	\end{itemize}
	Note that in the last case, when the non-linear variable $x$ has multiple occurrences with the same binder, the transition is non-reversible. Indeed, this means there exists another context $D_3 \neq D_2$ such that $D_2\ctxholep{x} = D_3\ctxholep{x}$. Then $\delta_a(q,\circlearrowleft) = \delta_a(q',\circlearrowleft)$ for some state $q' \neq q$, namely $q' = \mathsf{T}(D_1\ctxholep{\ttletin{\oc x = v}{D_3\ctxholep{\downred{x}}}} \,\lozenge_1\, \dots\, \lozenge_{\rk(a)},\;T)$.

	Working through the other cases reveals that the above subcase involving a $\mathtt{let}$-binding is the \emph{only one} that leads to a non-reversible transition in the tree-walking transducer. When the memory type $A$ is purely affine, the terms $u$ and $t_a$ ($a\in\Gamma$) do not contain $\mathtt{let}$-bindings; therefore, in that case, we get a reversible TWT.
  \end{proof}

\section{From the almost !-depth 1 IAM to invisible pebbles}\label{sec:depth-1}

Now that we have seen how to prove Theorem~\ref{thm:main-affine}, let us apply the same methodology to the almost !-depth 1 case, with another variation on the Interaction Abstract Machine.

The key challenge is that we can no longer rule out positions for the IAM token of the form $C\ctxholep{\ttletin{\oc x = \upblue{u}}{t}}$. If $x$ has multiple occurrences in $t$, we need some information to know which of these occurrences we should move to. The standard solution to this problem is to enrich IAM configurations with another data structure -- cf.\ the \enquote{boxes stack} of~\cite{DanosRegnierIAM} or the \enquote{log} of~\cite{AccattoliLV20}. In our simple low-depth case, a stack of variable occurrences will be enough.

\paragraph{The Almost !-Depth 1 IAM}

Let $v : o$ be an almost !-depth 1 term. To compute its normal form, we introduce a tree-generating machine whose configurations are of the form $(C\ctxholep{\downred{t}},T,L)$ or $(C\ctxholep{\upblue{t}},T,L)$, where $C\ctxholep{t} = v$, $T\in\{\resm,\resmtwo,\lpos\}^*$, and \emph{logs} $L$ and logged positions $l$ are defined by mutual induction as follows (please notice that $n\in\{0,1\}$ for !-depth 1 terms):
\[
\lpos\grameq (D_n,L_n) \qquad\qquad L_0\grameq \stempty \qquad L_n\grameq l\cdot L_{n-1}
\]
The initial configuration is $(\downred{v},\varepsilon,\varepsilon)$. To define the computation-step function, we start by reusing all the rules of the almost purely affine IAM (Definition~\ref{def:apaiam}), adapting them so that they do not change the log $L$, except the $\downred{x}$-rule for let-bound $x$, and the one for !-boxes. We then add the rules below, \emph{where $C_i$ ranges over contexts of depth $i\in\{0,1\}$} (as defined in \Cref{sec:lambda}) and $A\neq o$. Please notice that transition rules now depend also on the \emph{type} of the current subterm, indeed we have to distinguish between the ``almost'' and the ``depth-1'' exponentials. They are depicted in \Cref{fig:d2s-paiam}.

\begin{figure}
	{\small \begin{align*}
			(C_0\ctxholep{\ttletin{\oc x = u}{D_n\ctxholep{\downred{x^{!A}}}}},T,L_n) &\mapsto (C_0\ctxholep{\ttletin{\oc x = \downred{u^{!A}}}{D_n\ctxholep{x}}},(D_n,L_n) \cons T,\stempty)\\
			(C_0\ctxholep{\ttletin{\oc x = \upblue{u^{!A}}}{D_n\ctxholep{x}}},(D_n,L_n)\cons T,\stempty) &\mapsto (C_0\ctxholep{\ttletin{\oc x = u}{D_n\ctxholep{\upblue{x^{!A}}}}},T,L_n) \\
			(C\ctxholep{\ttletin{\oc x = u}{D_n\ctxholep{\downred{x^{!o}}}}},T,L_n\cons L) &\mapsto (C\ctxholep{\ttletin{\oc x = \downred{u^{!o}}}{D_n\ctxholep{x}}},T,L)\quad(\text{non-reversible})\\
			(C_0\ctxholep{\downred{\oc t^{!A}}},\lpos\cdot T,\stempty) &\mapsto (C_0\ctxholep{\oc\downred{t^{A}}},T,\lpos) \\ (C_0\ctxholep{\oc\upblue{t^{A}}},T,\lpos) &\mapsto (C_0\ctxholep{\upblue{\oc t^{!A}}},\lpos\cdot T,\stempty)\\
			(C\ctxholep{\downred{\oc t^{!o}}},T,L) &\mapsto (C\ctxholep{\oc\downred{t^{o}}},T,L) \quad(\text{non-reversible})
	\end{align*}}
	\caption{Computation-step function of the almost !-depth 1 IAM.}
	\label{fig:d2s-paiam}
\end{figure}

Again, this device, which is a just a specialization of the standard \IAM (but again non-reversible in order to handle linearly the almost affine terms),  successfully normalizes almost !-depth~1 terms of base type. Next, we would like to simulate it by some automaton model, to get a counterpart of Theorem~\ref{thm:main-affine} in this setting. The problem now is that both the log $L$ and the tape $T$ do not fit into the finite state of a tree-walking transducer, since their size cannot be statically bounded. Thus, we need to target a more powerful machine model.

\paragraph{Invisible Pebbles.}

Luckily, a suitable device has already been introduced by Engelfriet, Hoogeboom and Samwel~\cite{InvisiblePebbles}: the \emph{invisible pebble tree transducer} (IPTT).
Informally, it is a TWT extended with the ability to
put down pebbles on input nodes. The pebbles have colors that are taken in a
finite set. They can be later examined and removed: an IPTT can check whether
the \emph{last pebble to have been put down} is on the current position, and if
so, it can observe its color, and perhaps decide to remove it. The
\enquote{invisible} part means that only the last pebble can be seen. Thus, the
lifetimes of the pebbles follow a \emph{stack discipline} (last put down, first
removed). The number of pebbles used in a computation may be
unbounded.

\begin{defiC}[{\cite{InvisiblePebbles}}]\label{def:iptt}
  An \emph{invisible pebble tree transducer} $\Tree(\Gamma)\rightharpoonup\Tree(\Sigma)$ is made of:
  \begin{itemize}
  \item a finite set of \emph{states} $Q$ with an \emph{initial state} $q_0 \in Q$
  \item a finite set of \emph{colors} $\mathfrak{C}$
  \item a (partial) \emph{transition function} that sends tuples consisting of
  \begin{itemize}
    \item an input letter $a\in\Gamma$, a state $q\in Q$, a provenance $p \in \{\downarrow_\bullet, \circlearrowleft, \uparrow_1^\bullet, \dots,
    \uparrow_{\rk(a)}^\bullet\}$
    \item a boolean $\mathtt{isRoot}$ which must be false if $p=\;\downarrow_\bullet$
    \item a value $z$ which is either a color in $\mathfrak{C}$ or the symbol $\mathsf{None}$
  \end{itemize}
  to $\Tree(\Sigma,\; Q \times (\{\underbracket{\uparrow_\bullet}_{\mathclap{\text{prohibited if}\ \mathtt{isRoot}\ \text{is true}\quad}}, \circlearrowleft, \downarrow^\bullet_1, \dots,
  \downarrow^\bullet_k,\underbrace{\mathsf{remove}}_{\mathclap{\quad\text{prohibited
        if \(z=\mathsf{None}\)}}}\}\cup\{\mathsf{put}_{\mathfrak{c}} \mid \mathfrak{c}\in\mathfrak{C}\}))$
  \end{itemize}
\end{defiC}
Note that removing $z$ in the arguments and $\mathsf{remove/put}_{\mathfrak{c}}$ in the codomain would just yield an alternative presentation of tree-walking transducers (Definition~\ref{def:twt}).

The set of configurations of an invisible pebble tree transducer on an input tree $\tau$ is
\[ \qquad \underbrace{Q \times \{\downarrow_\bullet, \circlearrowleft, \uparrow_1^\bullet, \dots\} \times \{\text{nodes of}\ \tau\}}_{\text{TWT configuration}} \times \underbrace{(\mathfrak{C}\times\{\text{nodes of \(\tau\)}\})^*}_{\text{pebble stack}} \]

The transition function of the IPTT determines a computation-step function
by extending Definition~\ref{def:twt} in the expected way (the transducer stays at the same node after a
$\mathsf{remove}$ or $\mathsf{put}_{\mathfrak{c}}$ instruction).
Here is an example of a configuration over $\tau = a_1(b_2(c_3),c_4)$ for some
invisible pebble tree transducer: $(q, \downarrow_\bullet, c_3,
[(\mathfrak{a},b_2),(\mathfrak{b},c_3)])$. The top of the stack is the leftmost
element the list: it is an $\mathfrak{a}$-colored pebble on position $b_2$.
Since this differs from the current node $c_3$, the transducer does not see any
pebble ($z = \mathsf{None}$) even though there is a $\mathfrak{b}$-colored
pebble on $c_3$ further down the stack. If we execute the instruction
$\mathsf{put}_{\mathfrak{a}}$ while transitioning to state $q'$, we get the
configuration $(q', \circlearrowleft, c_3, [(\mathfrak{a},c_3),
(\mathfrak{a},b_2),(\mathfrak{b},c_3)])$. In that new configuration, the IPTT
now sees the topmost pebble ($z = \mathfrak{a}$) and is thus allowed to
$\mathsf{remove}$ it.

\begin{exa}
	We give a high-level explanation of an IPTT that converts binary to unary, e.g.\ $\mathtt{0}(\mathtt{1}(\mathtt{0}(\mathtt{1}(\varepsilon)))) \mapsto S(S(S(S(S(0)))))$.
	(From such a transducer, one can derive an IPTT for the function of Example~\ref{ex:lambda-bin2bin}, by replacing every pattern $S(x)$ by $a(x,x)$ and $0$ by $c$.)

	The states are $Q=\{q_0,q_1\}$ and the colors are $\mathfrak{C} = \{\mathfrak{a,b,c}\}$.
	We start in the initial state $q_0$ and travel down the input tree (i.e.\ left-to-right in the bitstring), ignoring the $\mathtt{0}$ nodes. When we encounter a $\mathtt{1}$ in state $q_0$, we put an $\mathfrak{a}$-colored pebble on the current position while moving down to the child and transitioning to the state $q_1$.

	The role of $q_1$ is then to output $2^n$ times $S$, where $n$ is the height of the current input subtree (indeed, we want to compute the positional value of the $\mathtt{1}$ bit we just saw). It does so by enumerating all colorings by $\{\mathfrak{b,c}\}$, in lexicographical order, of this subtree (excluding the input leaf $\varepsilon$ which stays uncolored).
	\begin{itemize}
		\item The initialization consists in moving down while putting $\mathfrak{b}$-pebbles on each successive $\mathtt{0/1}$ input node until we reach the leaf $\varepsilon$.
		\item Each time we are at $\varepsilon$ in state $q_1$, we output an $S$ node.
		\item We then backtrack while removing $\mathfrak{c}$-pebbles until we see a $\mathfrak{b}$-pebble. Then we remove it, put a $\mathfrak{c}$-pebble in its place, and move down while putting $\mathfrak{b}$-pebbles on all the nodes below, until we reach $\varepsilon$ again, etc. (This is analogous to incrementing a binary numeral.)
	\end{itemize}
	The backtracking phase of this last item may lead us to see an $\mathfrak{a}$-pebble, rather than a $\mathfrak{b}$-pebble, after having removed all $\mathfrak{c}$-pebbles. This means we have finished enumerating the colorings (and outputting an $S$ for each of them), and are back to the position where we transitioned from $q_0$ to $q_1$. Then we remove this $\mathfrak{a}$-pebble while transitioning back to $q_0$ and moving down to the child. The computation then proceeds as explained at the beginning; it ends when we reach $\varepsilon$ in state $q_0$, at which point we produce the leaf $0$.
\end{exa}

\paragraph{Simulating the IAM with a Single Stack.}

The IAM for almost affine depth-1 terms that we have described in Section~\ref{sec:depth-1} is not directly implementable on invisible pebble tree transducers, because of the presence of \emph{two}, possibly \emph{unbounded} stacks, the log $L$ and the tape $T$. In order to make the implementation possible, we tune the data structures in such a way that just one unbounded stack is needed (it is well-known in fact that an automaton/transducer featuring two unbounded stacks is Turing complete). We apply the following transformations:
\begin{itemize}
	\item We disentangle multiplicative and exponential information in the tape, as in~\cite{mackie_geometry_1995}. The multiplicative stack is then statically bounded by the size of types, as in the almost and purely affine cases.
	\item We merge the exponential part of the tape and the log into a \emph{unique} exponential stack. This is possible because there is at most one logged position at top level in the log, due to the depth-1 restriction.
	\item We add a boolean flag $\ans$ that indicates if the top of the exponential stack is actually the only log entry ($1$) or a tape entry ($0$).
\end{itemize}
The resulting machine is depicted in \Cref{fig:d-apaiam}.
\begin{figure}
	{\small \input{IAM-1-opt.tex}}
	\caption{Data structures and computation-step function of the single stack depth 1 APAIAM. Here $A\neq o$ and we write $\tm\esub{!\var}{\tmtwo}$ for $\ttletin{!\var=\tmtwo}{\tm}$.}
	\label{fig:d-apaiam}
\end{figure}
The fact that this machine is equivalent (strongly bisimilar) to the one defined in Section~\ref{sec:depth-1} is immediate. Then, we can use it to compile $\lambda$-transducers to IPTTs as in Section~\ref{sec:iam}:
\begin{lem}\label{lem:lambda-iptt}
  Almost !-depth 1 $\lambda$-transducer $\subseteq$ invisible pebble tree transducer $\equiv$ MSOT-S$^2$.
\end{lem}
Here, the equivalence between IPTT and MSOT-S$^2$ is a rephrasing of a result of
Engelfriet et al.~\cite[Theorem~53]{InvisiblePebbles}, as explained
in~\cite[\S3.3]{titosurvey}.

\section{Expressiveness of $\lambda$-transducers with preprocessing}
\label{sec:expr}

Now, let us prove Theorems~\ref{thm:msots} and~\ref{thm:main-depth-1}. We first
note that the left-to-right inclusions are immediate consequences of
Theorem~\ref{thm:main-affine} and Lemma~\ref{lem:lambda-iptt} combined with the following
facts:
\begin{itemize}
  \item MSOT-S (and, therefore, MSOT-S$^2$) are closed under precomposition by MSO
        relabeling (cf.~\cite[Section~3]{AttributedMSO} where MSOT-S are called
        \enquote{MSO term graph transductions});
  \item TWT $\subset$ MSOT-S (a slight variant
        of~\cite[Theorem~9]{AttributedMSO}, cf.~\cite[\S3.2]{titosurvey});
  \item TWT of linear growth $\subset$ MSOT (see
        e.g.~\cite[\S6.2]{EngelfrietIM21}), and purely affine
        $\lambda$-transducers have linear growth because the size of purely affine terms is non-increasing during $\beta$-reduction.
\end{itemize}
(\enquote{$f$ has linear growth} means that $|f(t)| = O(|t|)$.) Next, we turn to the converse inclusions.

\paragraph{MSOT $\subseteq$ Purely Affine $\lambda$-Transducer $\circ$ MSO Relabeling.}

We derive this from the results of Gallot, Lemay and
Salvati~\cite{LambdaTransducer,gallotPhD}. First, we introduce a slight
generalization of their machine model for MSOTs~\cite[\S2.3]{LambdaTransducer}.
Their model involves bottom-up regular lookahead, but as usual in automata
theory, this feature can be simulated by preprocessing by an MSO relabeling;
this is why we do not include it in our version.
\begin{defi}
  A \emph{GLS-transducer} $\Tree(\Gamma)\to\Tree(\Sigma)$ consists of:
  \begin{itemize}
  \item a finite set $Q$ of states, with a family $(A_q)_{q\in Q}$ of purely
    affine types;
  \item an initial state $q_0 \in Q$ and an output term $u : A_{q_0} \multimap
    o$ -- which, like all the terms $t$ below, may use constants $c : o^{\rk(c)}
    \multimap o$ for $c\in\Sigma$;
  \item for each $q\in Q$ and $a\in\Gamma$, a rule $q\langle
    a(x_1,\dots,x_{\rk(a)}) \rangle \to t\; q_1\langle x_1 \rangle\; \dots\;
    q_{\rk(a)}\langle x_{\rk(a)} \rangle$ where the $q_i$ and $t : A_{q_1}
    \multimap \dots \multimap A_{q_{\rk(a)}} \multimap A_q$ are chosen depending
    on $(q,a)$.
  \end{itemize}
\end{defi}
The semantics is that a GLS-transducer performs a top-down traversal $q_0\langle\tau\rangle\to^* \tau^\Downarrow$ of its
input tree $\tau$ which builds a $\lambda$-term $\tau^\Downarrow : A_{q_0}$. The
normal form of $u\;\tau^\Downarrow$ then encodes the output tree. Our model is a
bit more syntactically permissive than that of Gallot et al.\ (theirs would
correspond to using linear rather than affine terms, and forcing $u$ to be
$\lambda x.\; x$ -- so $A_{q_0} = o$);
therefore, it can compute at least everything that their model can:
\begin{thm}[from {\cite[Theorem~3]{LambdaTransducer}}]
  \label{thm:gls}
  MSOT $\subseteq$ GLS-transducer $\circ$ MSO relabeling.
\end{thm}
We derive our desired result on $\lambda$-transducers in two steps. \begin{itemize}
  \item Every GLS-transducer can be made \enquote{type-constant}: $\exists A : \forall q,\, A_q = A$. This uses an encoding trick, detailed in \Cref{app:gls}, that preserves pure affineness, but not linearity.
  \item A type-constant GLS-transducer can be turned into an MSO relabeling (that adds to each node its top-down propagated state) followed by a purely affine $\lambda$-transducer (which is just a GLS-transducer with $|Q|=1$).
\end{itemize}
From Theorem~\ref{thm:gls}, we thus get MSOT $\subseteq$ purely affine $\lambda\text{-transducer} \circ (\text{MSO relabeling})^2$, and since MSO relabelings are closed under composition~\cite[\S3]{AttributedMSO}, we are done.

\paragraph{MSOT-S $\subseteq$ Almost Purely Affine $\lambda$-Transducer $\circ$ MSO Relabeling.}

Following the same recipe as above, we reduce this to Gallot et al.'s characterization of \mbox{MSOT-S}~\cite{LambdaTransducer}. Since they use Kanazawa's almost linear $\lambda$-terms~\cite{KanazawaJournal} (which we discussed in the introduction), we need to translate such terms into our almost purely affine terms. Let us introduce the abbreviation $\lambda!x.\,t = (\lambda y.\, \ttletin{\oc x = y}{t})$ where $y$ is a fresh variable. We define inductively:
$\wn c = \lambda!x_1.\, \dots\, \lambda!x_{\rk(c)}.\, \oc(c\, x_1\, \dots\, x_{\rk(c)})$ for constants $c$ in a ranked alphabet $\Sigma$, and
\[ \wn x = \begin{cases}
  \oc x &\text{if}\ x : o\\
  x &\text{otherwise}
\end{cases} \qquad \wn(\lambda x.\,t) = \begin{cases}
  \lambda \oc x.\, \wn t &\text{if}\ x : o\\
  \lambda x.\, \wn x  &\text{otherwise}
\end{cases} \qquad \wn(t\,u) = (\wn t)\, (\wn u)
\]
\begin{clm}
  Let $t : A$ be almost affine as defined in~\cite{AlmostAffine}. Then $\wn t$ is almost purely affine, with type $A\{o:=\oc o\}$. When $t : o$, if $t \longrightarrow_\beta^*\widetilde\tau$ (for some $\tau\in\Tree(\Sigma)$) then $\wn t \longrightarrow_\beta^* \oc\widetilde\tau$.
\end{clm}
The inductive rule for application implies that $\wn(u\;\widehat\tau((t_a)_{a\in\Gamma})) = (\wn u)\; \widehat\tau((\wn t_a)_{a\in\Gamma})$, so the above claim allows us to translate $\lambda$-transducers using almost affine terms \textit{à la} Kanazawa to $\lambda$-transducers using our notion of almost purely affine terms. Finally, to bridge the gap with~\cite[Theorem~3]{LambdaTransducer}, we observe that the trick of \Cref{app:gls} to turn GLS-transducers into $\lambda$-transducers still works for almost affine terms in the sense of~\cite{AlmostAffine,KanazawaJournal}.

\paragraph{MSOT-S$^2$ $\subseteq$ Almost !-Depth 1 $\lambda$-Transducer $\circ$ Relabeling.}

Having just finished proving Theorem~\ref{thm:msots} above, we may use it right away:
\begin{align*}
  \text{MSOT-S$^2$} &\equiv \text{almost purely affine $\lambda$-transducer} \circ \text{MSO relabeling}\circ \text{MSOT-S}\ (\text{Thm.~\ref{thm:msots}}) \\
              &\equiv \text{almost purely affine $\lambda$-transducer} \circ \text{MSOT-S}\qquad\qquad\qquad (*)\\
                 &\equiv (\text{almost purely affine $\lambda$-transducer})^2 \circ \text{MSO relabeling} \qquad (\text{Thm.~\ref{thm:msots}})
\end{align*}
The line $(*)$ above relies on the fact that MSOT-S are closed under \emph{post}composition by MSO relabelings, cf.~\cite[\S3.3]{titosurvey}.
To conclude, we apply the composition property of $\lambda$-transducers (Proposition~\ref{prop:compo}), noting that if $A$ and $B$ are almost purely affine types, then $A\{o:=B\}$ is almost !-depth 1 (indeed, every `!' in it is applied to either $o$ or $B$).

\section{Conclusion}

In this paper, we established several expressivity results relating a typed $\lambda$-calculus to tree transducers. This can be seen as furthering \titocecilia's \enquote{implicit automata} research programme~\cite{iatlc1}, even though the formal setting is slightly different; indeed, we settle one of their conjectures in Corollary~\ref{cor:inexpress}. From a purely automata-theoretic perspective, our characterization of MSOT-S$^2$ is the first that involves a \enquote{one-way} device, performing a single bottom-up pass on its input (modulo preprocessing).

The equivalences between \enquote{one-way} $\lambda$-transducers and
tree-walking / invisible pebble tree transducers can be seen as a trade-off
between a sophisticated memory (higher-order data) and freedom of movement on
the input (tree-walking reading head). This is arguably a sort of qualitative
space/time trade-off (more movement means more computation steps). This is
similar to the reasons that led the Geometry of Interaction to be used in
implicit computational complexity when dealing with sublinear space complexity classes ---
an application area pioneered by Schöpp~\cite{DBLP:conf/csl/Schopp06,bllspace}
and leading to several further
works~\cite{dal_lago_computation_2016,Parsimonious}. These successes even led to
the belief that the GoI should give a reasonable space cost model, that is to
say comparable with the one of Turing machines; but this belief is now known to
be wrong in the general case of the untyped $\lambda$-calculus~\cite{goispace}.

\paragraph{More Related Work.}

Katsumata~\cite{Katsumata08} has connected a categorical
version of the GoI (the \enquote{Int-construction}) to attribute grammars~\cite{AttributeGrammars}, which are \enquote{essentially [a] notational variation} on tree-walking transducers (quoting Courcelle \& Engelfriet~\cite[\S8.7]{courcellebook}). Recently, Pradic and Price~\cite{entics:14804} have used a \enquote{planar} version of this categorical GoI in order to prove an \enquote{implicit automata} theorem. Further GoI-automata connections of this kind are discussed in~\cite[\S1.1]{nguyen2023twoway}.

Our methodology of connecting $\lambda$-calculus and automata via abstract machines may be compared to Salvati and Walukiewicz's~\cite{SalvatiWalukiewiczCPDA} use of Krivine machines in the theory of higher-order recursion schemes. Clairambault and Murawski~\cite{MAHORS} also compile affine recursion schemes to automata using a game semantics that can be seen as a denotational counterpart of the operational Interaction Abstract Machine. Ghica
exploited ideas from the GoI and game semantics, to design a compiler from a higher-order functional language directly to
digital circuits~\cite{ghica_geometry_2007}, in particular targeting Mealy machines.

\begin{rem}\label{rem:additives}
  The aforementioned work~\cite{MAHORS} is a rare example of application of some GoI variant to a setting that features the \emph{additive connectives} of linear logic.
  The additives are a source of non-negligible complications in the GoI, motivating our choice (discussed in the introduction) to follow the approach of Gallot et al.~\cite{LambdaTransducer,gallotPhD} rather than \titocecilia's \enquote{implicit automata}~\cite{iatlc1,titoPhD}. Indeed, the latter's solution to overcome the limitations evidenced by Corollary~\ref{cor:inexpress} is to work with a linear
  $\lambda$-calculus with additive connectives, enabling more flexible linear usage patterns. This is analogous (see~\cite[Remark~6.0.1]{titoPhD} for an actual technical
  connection) to moving from \enquote{strongly single use} to \enquote{single use} macro tree
  transducers~\cite[Section~5]{MacroMSO}.
  That said, later work by the first author with Dartois and Peyrat~\cite{DTP26} demonstrates that the ideas from~\cite{MAHORS} may be meaningfully applied to tree transducers.
\end{rem}

Salvati has shown~\cite{SalvatiTWA} that the string languages defined by
abstract categorial grammars, which are very close to our purely linear
$\lambda$-transducers, coincide with the output languages of tree-to-string
tree-walking transducers. He explains in his habilitation
thesis~\cite[\S3.2]{SalvatiHDR} that the proof ideas are similar to a game
semantics of multiplicative linear logic --- and the latter is closely related
to GoI, as mentioned above. It would be interesting to understand to which
extent his approach implicitly resembles ours, despite a very different
presentation.

We also note that in the same paper that introduces almost linear
$\lambda$-terms~\cite{KanazawaJournal}, Kanazawa studies a notion of
\enquote{links in typed $\lambda$-terms} that looks like a form of GoI. However,
these links are only well-behaved for $\lambda$-terms in normal form, while the
Interaction Abstract Machine does not have this drawback.
Finally, let us stress that our use of the IAM has the advantage, compared to
the aforementioned works~\cite{Katsumata08,MAHORS,SalvatiTWA,KanazawaMSO}, of
adapting to the presence of the exponential modality `$!$'. This is crucial in
our proof of Theorem~\ref{thm:main-depth-1}.

\paragraph{Perspectives.}

The obvious direction for further work is to study the $\text{MSOT-S}^{k+1}$ and almost !-depth $k$ hierarchies for $k \geq 2$. While the argument at the end of \Cref{sec:expr} easily generalizes to show that the former is included in the latter, we have no reason to believe that they coincide.
As a more modest conjecture (\enquote{!-depth 1} means \enquote{no nested `!'s}):
\begin{conj}
  !-depth 1 $\lambda$-transducer $\circ$ MSO relabeling $\equiv \text{MSOT} \circ \text{MSOT-S}$.
\end{conj}
We believe that the \emph{reversible} tree-walking transducers that we have introduced also deserve to be studied further. Indeed, we expect that they should be closed under composition (cf.~\cite{ReversibleTransducers} over strings) and verify the \enquote{single-use restriction} of~\cite[\S8.2]{courcellebook}; the latter would imply that they can be translated into MSO transductions.

\section*{Acknowledgment}
  We thank Damiano Mazza and Cécilia Pradic for discussions on the possible applications of the Geometry of Interaction to implicit complexity and automata theory.
  
  \emph{Funding:} the first author was supported by the DyVerSe project (ANR-19-CE48-0010). The second author was supported by the ANR PPS Project (ANR-19-CE48-0014) and the European Union's Horizon 2020 research and innovation programme under the Marie Skłodowska-Curie grant agreement No.~101034255.

\bibliographystyle{alphaurl}
\bibliography{bibliography}

\appendix

\section{Proof of Theorem~\ref{thm:paiam-soundness}}\label{sec:paiam-soundness}
We prove in this section the soundness of the PAIAM. The idea is simple. We show that (i) what is computed by the PAIAM is invariant under $\beta$ reduction, and (ii) that the computation on normal forms is correct.

We use the notion of improvement from~\cite{AccattoliLV20}.
\begin{defi}[Improvements]\label{def:impr}
	Given two abstract machines with sets of states $Q$ and $S$, a relation $\mathcal{R}\subseteq Q\times S$ is
	an \emph{improvement} if
	given $(q,s)\in\mathcal{R}$ the following conditions hold:
	\begin{enumerate}
		\item \emph{Final state right}: if $s\in S$ is a final state, then $q\rightarrow^n q'$, for some final state
		$q'\in Q$ and $n\geq 0$.
		\item \emph{Transition left}: if $q\in Q$ and $q\rightarrow q'$, then there exists $q''\in Q$ and $s'\in S$ such
		that $q'\rightarrow^m q''$, $s\rightarrow^n s'$,
		$q''\mathcal{R}s'$ and $n\leq m+1$.
		\item \emph{Transition right}: if $s\in S$ and $s\rightarrow s'$, then there exists $q'\in Q$ and $s''\in S$ such that
		$q\rightarrow^m q'$, $s'\rightarrow^n s''$, $q'\mathcal{R}s''$ and
		$m\geq n+1$.
	\end{enumerate}
\end{defi}
What improves along an improvement is the number of transitions required to reach a final state, if any. Diagrammatically, we have the following situation:
\begin{center}
	\begin{tabular}{cc|ccc}
		\begin{tikzpicture}[node distance=30mm, auto, transform
			shape,scale=1]
			\node (p) at (0,0) {$s$};
			\node (q) at (3,0) {$q$};
			\node (w) at (0,-.8) {$s'$};
			\node (t) at (0,-1.6) {$s''$};
			\node (r) at (3,-1.6) {$q'$};
			\node at (1.5,0) {$\mathcal{R}$};
			\node at (1.5,-1.6) {$\mathcal{R}$};
			\draw (p) edge[->] node {} (w);
			\draw (w) edge[->,dashed] node {$^m$} (t);
			\draw (q) edge[->,left, dashed] node {$^{n\leq m+1}$} (r);
		\end{tikzpicture}
		&&&
		\begin{tikzpicture}[node distance=30mm, auto, transform
			shape,scale=1]
			\node (p) at (0,0) {$s$};
			\node (q) at (3,0) {$q$};
			\node (w) at (0,-1.6) {$s'$};
			\node (t) at (3,-.8) {$q'$};
			\node (r) at (3,-1.6) {$q''$};
			\node at (1.5,0) {$\mathcal{R}$};
			\node at (1.5,-1.6) {$\mathcal{R}$};
			\draw (p) edge[->,dashed] node {$^{m\geq n+1}$} (w);
			\draw (q) edge[->] node {} (t);
			\draw (t) edge[->,left,dashed] node {$^n$} (r);
		\end{tikzpicture}
	\end{tabular}
\end{center}
\begin{propC}[\cite{AccattoliLV20}]\label{prop:imp}
	Let $\mathcal{R}\subseteq Q\times S$ be an improvement, and $q\mathcal{R}s$. Then:
	\begin{enumerate}
		\item \label{prop:termination}
		\emph{(Non) Termination:} $q\to^*q'$ with $q'$ final if and only if $s\to^*s'$ with $s'$ final.

		\item \label{lemma:improv}
		\emph{Improvement:} $|q|\geq|s|$, where $|a|$ is the length of the run from $a$ to a final state, if any, and $+\infty$ otherwise.
	\end{enumerate}
\end{propC}
In order to formally define improvements, we need a formal inductive definition of $\Tree(\Sigma,\mathcal{K})$.
\begin{defi} $\Tree(\Sigma,\mathcal{K})$ is the smallest set inductively defined by the following rules:
	\[
	\infer{\kappa\in\Tree(\Sigma,\mathcal{K})}{\kappa\in\mathcal{K}} \qquad\qquad
	\infer{c(\tau_1, \ldots, \tau_n)\in\Tree(\Sigma,\mathcal{K})}{\tau_1, \ldots, \tau_n\in\Tree(\Sigma,\mathcal{K})\qquad n\geq 0\qquad c\in\Sigma}
	\]
\end{defi}
\newcommand{\relim}{\vartriangleright}
Then, we are able to define the relation $\relim^*\,\subseteq\Tree(\Sigma,\mathcal{I})\times\Tree(\Sigma,\mathcal{I})$, where $\mathcal{I}$ is the set of PAIAM configurations, as the reflexive and transitive closure of $\relim$, defined below. The interesting cases are the ones for PAIAM configurations (we overload the symbol $\relim$):
\begin{center}$\small\begin{array}{c@{\hspace{.8cm}}c@{\hspace{.8cm}}c}
		\multicolumn{3}{c}{\text{Reduction on Contexts}}\\[.17cm]
		\infer{\tm\ctx \tob \tmtwo \ctx}{\tm \tob \tmtwo}
		&
		\infer{\ctx\tm \tob \ctx\tmtwo}{\tm \tob \tmtwo}
		&
		\infer{\ctxtwop{\ctx  } \tob \ctxtwop{\ctx'}}{\ctx \tob \ctx'}\\[.17cm]
		\multicolumn{3}{c}{\text{Relation of Positions}}\\[.17cm]
		\infer[\textsf{rdx}]{(\tm,\ctx)\relim(\tmtwo,\ctx)}{\tm\tob\tmtwo}
		&
		\infer[\textsf{ctx}]{(\tm,\ctx)\relim(\tm,\ctxtwo)}{\ctx\tob\ctxtwo} &
		\infer[\textsf{sub}]{(\tm,\ctxp{(\la{\var}\ctxtwo )\tmtwo}) \relim (\tm\isub\var\tmtwo,\ctxp{\ctxtwo\isub\var\tmtwo}) }{}
		\\[.17cm]
		\multicolumn{3}{c}{\infer[\textsf{sub2}]{(u,\ctxp{(\la\var \ctxthreep{\var})\ctxtwo})\relim (u,\ctxp{\ctxthreep{\ctxtwo}})}{}}\\[.17cm]
		\multicolumn{3}{c}{\text{Relation of Tapes}}\\[.17cm]
		\infer[\textsf{tok1}]{\epsilon\relim \epsilon}{}
		&
		\infer[\textsf{tok2}]{\resm\cdot\tape\relim
			\resm\cdot\tape'}{\tape\relim\tape'}
		&
		\infer[\textsf{tok3}]{\resmtwo\cdot\tape\relim
			\resmtwo\cdot\tape'}{\tape\relim\tape'}
	\end{array}$\end{center}
\begin{center}$\small\begin{array}{c@{\hspace{.8cm}}c}
		\multicolumn{2}{c}{\text{Relation of Trees}}\\[.17cm]
		\multicolumn{2}{c}{
			\infer[\textsf{configuration}]{(\tm,\ctx,\tape,\pol)\relim
				(\tmtwo,\ctxtwo,\tape',\pol')}{(\tm,\ctx)\relim
				(\tmtwo,\ctxtwo)\qquad\tape\relim\tape'\qquad
				\pol=\pol'}
		}\\[.17cm]
		\infer[\textsf{node}]{c(\tau_1, \ldots, \tau_n)\relim c(\sigma_1, \ldots, \sigma_n)}{\tau_i\relim\sigma_i\qquad n\geq 0}
	\end{array}$\end{center}
\begin{prop}$\relim^*$ is an improvement.
\end{prop}
\begin{proof}
	The proof is routine and amounts to check that all the diagrams close correctly. See \cite{AccattoliLV20,Parsimonious} for more details.
\end{proof}
\paragraph{Statement of Theorem~\ref{thm:paiam-soundness}, Soundness of the PAIAM}

For any purely affine term $v : o$, the output of $\paiam(v)$ is the unique
$\tau\in\Tree(\Sigma)$ such that $\widetilde\tau$ is the normal form of $v$.

\begin{proof}
	Clearly, $(\red{\underline{v}},\epsilon)\relim^*(\red{\underline{\widehat{\tau}}},\epsilon)$, because $v\tob^*\widehat{\tau}$. Then, by the properties of improvements, also the final states of the runs starting from $(\red{\underline{v}},\epsilon)$ and $(\red{\underline{\widehat{\tau}}},\epsilon)$, respectively are related by $\relim^*$. Since by Prop.~\ref{prop:paiam-on-normal-form}, $(\red{\underline{\widehat{\tau}}},\epsilon)\rightsquigarrow^*\tau$, then the also $(\red{\underline{v}},\epsilon)\rightsquigarrow^*\tau$. In fact, only the rule \textsf{node} for $\relim$ can be used to relate tree constructors.
\end{proof}
Proposition~\ref{prop:apaiam-soundness} is proved analogously, adding the following clauses to the ones already defined above to build the improvement $\relim$:
\begin{center}$\small\begin{array}{c@{\hspace{.8cm}}c@{\hspace{.8cm}}}
		\multicolumn{2}{c}{\text{Reduction on Contexts}}\\[.17cm]
		\infer{\ttletin{x = C}{t} \tob \ttletin{x = C}{u}}{\tm \tob \tmtwo}
		&
		\infer{\ttletin{x = t}{C} \tob \ttletin{x = u}{C}}{\tm \tob \tmtwo}\\[.17cm]
		\multicolumn{2}{c}{\text{Relation of Positions}}\\[.17cm]
		\multicolumn{2}{c}{\infer[\textsf{sub3}]{(\tm,\ctxp{\ttletin{x=L\ctxholep{!u}}{\ctxtwo}}) \relim (\tm\isub\var\tmtwo,\ctxp{L\ctxholep{\ctxtwo\isub\var\tmtwo}})}{}}
		\\[.17cm]
		\multicolumn{2}{c}{\infer[\textsf{sub4}]{(u,\ctxp{\ttletin{x= L\ctxholep{!D}}{\ctxthreep{\var}}})\relim (u,\ctxp{L\ctxholep{\ctxthreep{\ctxtwo}\isub{\var}{\ctxtwop{u}}}})}{}}\\[.17cm]
	\end{array}$\end{center}
The proof of soundness of the almost depth-1 IAM follows the same pattern. In particular, the machine is a specialization to almost depth-1 terms of the machine in~\cite{AccattoliLV20}. The soundness thus follows as a corollary.

\section{Making GLS-transducers type-constant}
\label{app:gls}

The problem: a GLS-transducer involves a family
of types $A_q$ indexed by states. The return type of a subtree depends on the
top-down state propagated to its root.
With our $\lambda$-transducers that have a single memory type, we would
therefore like to use something like $\displaystyle\bigoplus_{q\in Q} A_q$ but we do not have
access to additive connectives. Instead, let
\begin{align*}
  A &= B_{q_1,1} \multimap \dots \multimap B_{q_1,k[q_1]} \multimap B_{q_2,1} \multimap \dots
  \multimap B_{q_{|Q|},k[q_{|Q|}]} \multimap o\\
  \text{where}\ A_q &= B_{q,1} \multimap \dots
\multimap B_{q,k[q]} \multimap o
\end{align*}
using an arbitrary enumeration $Q=\{q_1,\dots,q_{|Q|}\}$ of states.

We shall now define conversion functions $\iota_q : A_q \multimap A$ and
$\mathtt{cast}_q : A \multimap A_q$ such that for any closed $\lambda$-term $t :
A_q$, we have $\mathtt{cast}_q\; (\iota_q\; t) =_{\beta\eta} t$. Since the
output alphabet contains at least one symbol $\ell$ of rank $0$, which is given
to us as a constant of type $o$ in the $\lambda$-calculus, one can show that
every type $C$ is inhabited by some closed term $\mathtt{dummy}[C]$: for $C = D_1
\multimap \dots \multimap D_n \multimap o$, take $\mathtt{dummy}[C] = \lambda x_1
\dots x_n.\; \ell$. We then take
\begin{align*}
	\iota_{q_1} &= \lambda z.\; \lambda x_1 \dots x_N.\; z\; x_1\; \dots \; x_{k[q_1]}
	\quad\text{where}\ N = k[q_1] + \dots + k[q_{|Q|}]\\
	\mathtt{cast}_{q_1} &= \lambda y.\; \lambda x_1 \dots x_{k[q_1]}.\;
	y\; x_1\; \dots \; x_{k[q_1]} \; \mathtt{dummy}[B_{q_2,1}] \;\dots\; \mathtt{dummy}[B_{q_{|Q|},k[q_{|Q|}]}]
\end{align*}
Similar definitions work for each state $q_2,\dots,q_{|Q|}$.

Note that these terms $\iota_q$ and $\mathtt{cast}_q$ are affine but not linear!
\end{document}